\pdfoutput=1
\documentclass[a4paper,UKenglish,cleveref, autoref, thm-restate]{lipics-v2019}

\title{Settling the relationship between Wilber's bounds for dynamic optimality} 

\author{Victor Lecomte}{Columbia University, USA}{vl2414@columbia.edu}{}{Research supported by a fellowship of the Belgian American Educational Foundation.}

\author{Omri Weinstein}{Columbia University, USA}{omri@cs.columbia.edu}{}{Research supported by NSF CAREER award CCF-1844887.}

\authorrunning{V. Lecomte and O. Weinstein} 

\Copyright{Victor Lecomte and Omri Weinstein} 



\newif\ifarxiv
\arxivtrue

\ifarxiv
    \relatedversion{This is the full version of the paper. The version published in the proceedings of ESA 2020 doesn't contain Section~\ref{sec:speculation}. The versions are otherwise identical.}
\else
    \relatedversion{A full version of the paper is available on arXiv \cite{fullversion} at \url{https://arxiv.org/abs/1912.02858}. This conference version doesn't contain the last section, which relates the \FB{} to the Independent Rectangle bound. The versions are otherwise identical.}
\fi

\acknowledgements{We want to thank the anonymous reviewers for their enthusiastic feedback and the numerous typos they spotted.}

\nolinenumbers 

\ifarxiv
    \hideLIPIcs  
\fi

\EventEditors{Fabrizio Grandoni, Peter Sanders, and Grzegorz Herman}
\EventNoEds{3}
\EventLongTitle{28th Annual European Symposium on Algorithms (ESA 2020)}
\EventShortTitle{ESA 2020}
\EventAcronym{ESA}
\EventYear{2020}
\EventDate{September 7--9, 2020}
\EventLocation{Pisa, Italy (Virtual Conference)}
\EventLogo{}
\SeriesVolume{173}
\ArticleNo{72}

\usepackage{tikz}

\usetikzlibrary{angles,quotes,decorations.pathreplacing,calc,decorations.markings,shapes.misc,positioning,patterns}
\tikzstyle{dot} = [circle, fill, minimum size=.1cm, inner sep=0]
\tikzset{cross/.style={cross out, draw, 
         minimum size=2*(#1-\pgflinewidth), 
         inner sep=0pt, outer sep=0pt},
         cross/.default={.9mm}}
\tikzset{every picture/.style={thick, font=\small, scale=0.8}}

\usepackage{mathtools,tabu,booktabs,pifont}
\newtheorem{fact}[theorem]{Fact}
\newtheorem{algorithm}[theorem]{Algorithm}

\definecolor{DarkRed}{rgb}{0.5,0.1,0.1}
\definecolor{DarkBlue}{rgb}{0.1,0.1,0.5}
\definecolor{DarkGreen}{rgb}{0.1,0.5,0.1}
\newcommand{\cmark}{\textcolor{DarkGreen}{\ding{51}}}
\newcommand{\xmark}{\textcolor{DarkRed}{\ding{55}}}

\tikzstyle{point} = [circle, fill, minimum size=.08cm, inner sep=0]
\newcommand{\upRect}{
    \begin{tikzpicture}[scale=0.2]
    \draw (0,0) -- (1,0) -- (1,1) -- (0,1) -- cycle;
    \draw (0,0) -- (1,1);
    \node[point] at (0,0) {};
    \node[point] at (1,1) {};
    \end{tikzpicture}
}
\newcommand{\downRect}{
    \begin{tikzpicture}[scale=0.2]
    \draw (0,0) -- (1,0) -- (1,1) -- (0,1) -- cycle;
    \draw (0,1) -- (1,0);
    \node[point] at (0,1) {};
    \node[point] at (1,0) {};
    \end{tikzpicture}
}

\newcommand{\FB}{Funnel bound}
\newcommand{\IB}{Alternation bound}
\newcommand{\zr}{z-rectangle}

\renewcommand{\th}{^\text{th}}

\newcommand{\blocks}{\mathrm{blocks}}
\newcommand{\mix}{\mathrm{mix}}
\newcommand{\mixValue}{\mathrm{mixValue}}
\newcommand{\ttL}{\mathtt{L}}
\newcommand{\ttR}{\mathtt{R}}
\newcommand{\TL}{{\cT_\ttL}}
\newcommand{\TR}{{\cT_\ttR}}
\newcommand{\xL}{{\x_\ttL}}
\newcommand{\xR}{{\x_\ttR}}
\newcommand{\PL}{{P_\ttL}}
\newcommand{\PR}{{P_\ttR}}
\newcommand{\FL}{{F_\ttL(P,p)}}
\newcommand{\FR}{{F_\ttR(P,p)}}

\newcommand{\OPT}{\mathsf{OPT}}
\newcommand{\IRB}{\mathsf{IRB}}
\newcommand{\upIRB}{\mathsf{IRB}_{\upRect}}
\newcommand{\downIRB}{\mathsf{IRB}_{\downRect}}
\newcommand{\addUp}{\mathrm{add}_{\upRect}}
\newcommand{\addDown}{\mathrm{add}_{\downRect}}

\newcommand{\ALT}{\mathsf{Alt}}
\newcommand{\W}{\mathsf{Funnel}}
\newcommand{\Funnel}{\W}
\newcommand{\zRects}{\mathsf{zRects}}
\newcommand{\spl}{\mathrm{split}}
\newcommand{\rot}[1]{#1^{\perp}}
\newcommand{\rotrot}[1]{#1^{\perp\perp}}
\newcommand{\rect}[1]{\square #1}

\newcommand{\rev}[1]{\overline{#1}}

\newcommand{\R}{\mathbb{R}}
\DeclarePairedDelimiter\floor{\lfloor}{\rfloor}

\newcommand{\x}{X}
\newcommand{\xTilde}{\tilde{\x}}

\newcommand{\cT}{\mathcal{T}}

\begin{document}

\maketitle

\begin{abstract}
In FOCS 1986, Wilber proposed two combinatorial lower bounds  
on the operational cost of any binary search tree (BST) for a given access sequence $\x \in [n]^m$.  
Both bounds play a central role in the ongoing pursuit of the \emph{dynamic optimality conjecture} 
(Sleator and Tarjan, 1985), but their relationship remained unknown for more than three decades. We show that 
Wilber's \emph{\FB} dominates his \emph{\IB} for all $\x$, and give a tight 
$\Theta(\lg\lg n)$ separation for some $\x$, answering Wilber's conjecture and an open problem of Iacono, Demaine et. al. 
The main ingredient of the proof is a new \emph{symmetric} characterization of Wilber's \FB{}, which proves that 
it is invariant under \emph{rotations} of $\x$. 
We use this characterization to provide initial indication that the \FB{} matches the Independent Rectangle 
bound (Demaine et al., 2009), 
by proving that when the Funnel bound is constant, $\upIRB$ is linear.
To the best of our knowledge, our results provide the first progress on Wilber's conjecture 
that the \FB{} is dynamically optimal (1986).
\end{abstract}


\section{Introduction}

The \emph{dynamic optimality conjecture} of Sleator and Tarjan \cite{ST85} postulates the existence of 
an \emph{instance optimal} binary search tree algorithm (BST), namely, an online self-adjusting BST
whose running time\footnote{i.e. the number of pointer movements and tree rotations performed  by the BST} 
matches the best possible running time \emph{in hindsight} for any fixed sequence of queries.
More formally, letting $\cT(\x)$ denote the operational time of a BST algorithm $\cT$ on a sequence $\x= (x_1,\ldots,x_m) \in [n]^m$ 
of keys to be searched, the conjecture says that there is an online BST $\cT$ such that
$\forall\x$, $\cT(\x) \leq O(\OPT(\x))$, where $\OPT(\x) := \min_{\cT'} \cT'(\x)$ denotes the optimal 
offline cost for $X$.
Such instance optimal algorithms are generally impossible, as an offline algorithm that sees the input $\x$ in advance 
can simply ``store the answers'' and output them in $O(1)$ per operation, which is why worst-case analysis is the typical 
benchmark for online algorithms. Nevertheless, 
in the BST model, where the competing class of algorithms are self-adjusting binary search trees,  instance optimality is 
an intriguing possibility.
After 35 years of active research, two BST algorithms are still conjectured to be constant-competitive: 
The first one is the celebrated \emph{splay tree} of \cite{ST85},  
the second one is the more recent \emph{GreedyFuture} algorithm \cite{Luc88,DHIKP09,Munro00}.
However, optimality of both splay trees and GreedyFuture was proven only in special cases, 
 and they are not known to be $o(\lg n)$-competitive for general access sequences $\x$ 
(note that every balanced BST is trivially $O(\lg n)$-competitive).
The best provable result to date on the algorithmic side is an $O(\lg\lg n)$-competitive BST, 
the \emph{Tango Tree} (\cite{DHIKP09} and its subsequent variants \cite{Multisplay,BDDF10}).

The ongoing pursuit of dynamically-optimal BSTs motivated the development of lower bounds on the 
cost of the offline solution $\OPT(\x)$, attempting to capture the ``correct'' complexity measure of a fixed access 
sequence $\x$ in the BST model,  and thereby providing a concrete benchmark for competitive analysis.
Indeed, one defining feature of the dynamic optimality problem (and the reason why it is a viable possibility)
is the existence of nontrivial lower bounds on $\OPT(\x)$ for individual \emph{fixed} access sequences $\x$, as opposed 
to distributional lower bounds.
\footnote{For example, Wilber's \IB{} can be 
used to show that the ``bit-reversal'' access sequence obtained by reversing the binary representation of the 
monotone sequence $\{1,2,3,\ldots, n\}$ has cost $\Omega(\lg n)$ per operation \cite{Wil89}.}  
These lower bounds are all derived from a natural geometric interpretation of the access sequence $X = x_1,\ldots, x_m$
as a point set on the plane, mapping the $i\th$ access $x_i$ to point $(x_i,i)$ (\cite{DHIKP09, IaconoSurvey}, 
see Figure~\ref{fig:geom-view}). The earliest lower bounds on $\OPT(\x)$ were proposed in an influential paper 
of Wilber \cite{Wil89}, and are the main subject of this paper.

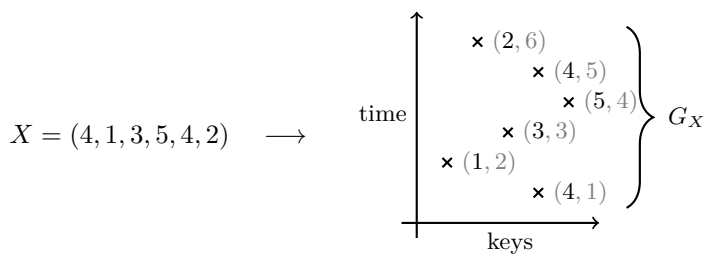
\begin{figure}[h]
\centering
\newcommand{\axes}[2]{
    \newcommand{\start}{-.5}
    \draw[->] (\start,0) -- (#1,0) node[pos=0.55, below] {keys};
    \draw[->] (0,\start) -- (0,#2) node[pos=0.55, left] {time};
}
\newcommand{\pointAndCoord}[2]{
\node[cross] (p1) at (#1,#2) {};
\node[gray, right=0mm of p1] {$({\color{black}{{#1}}},#2)$};
}
$
X = (4,1,3,5,4,2)
\quad \longrightarrow \quad
\vcenter{\hbox{
\begin{tikzpicture}[scale=0.5]
\axes{6}{7}
\pointAndCoord{4}{1}
\pointAndCoord{1}{2}
\pointAndCoord{3}{3}
\pointAndCoord{5}{4}
\pointAndCoord{4}{5}
\pointAndCoord{2}{6}
\draw [decorate,decoration={brace,amplitude=10pt},xshift=9mm,yshift=0pt]
(6,6.5) -- (6,0.5) node [midway,xshift=8mm] 
{$G_X$};
\end{tikzpicture}
}}
$
\caption{Transforming $X$ into its geometric view $G_X$}\label{fig:geom-view}
\end{figure}

\paragraph*{The \IB{}} Wilber's first lower bound, the \emph{\IB} $\ALT_\cT(\x)$, 
counts the total number of left/right alternations obtained by searching the keys $\x=(x_1,\ldots, x_m)$ 
on a \emph{fixed} (static) binary search tree $\cT$, where alternations are summed up over all nodes 
$v\in \cT$ of the ``reference tree'' $\cT$ (see Figure \ref{fig:alt} and the formal definition in Section~\ref{sec:prelims}). 
Thus, the \IB{} is actually a family of lower bounds, optimized by the choice of the 
reference tree $\cT$, and we henceforth define $\ALT(\x) \coloneqq \max_\cT \ALT_\cT(\x)$.
This lower bound played a key role in the design and analysis of 
Tango trees and their variants \cite{DHIP07, Multisplay}, whose operational cost is in fact shown to be 
$O(\lg\lg n)\cdot \ALT_\cT(\x) \leq O(\lg\lg n)\cdot \OPT(\x)$ (when setting the reference tree $\cT$ to be 
the canonical balanced BST on $[n]$). Unfortunately, this bound is not tight,  
as we show that there are access sequences $\xTilde$ for which
$\ALT_\cT(\xTilde) \leq O(\OPT(\xTilde)/\lg \lg n)$   
\emph{simultaneously} for all choices of reference trees $\cT$ 
(previously, this was known only for any \emph{fixed} $\cT$ \cite{IaconoSurvey}), 
and hence the combined bound $\ALT(\x)$ does not capture dynamic 
optimality in general. 
Nevertheless, the algorithmic interpretation of the \IB{} is an interesting proof-of-concept of how 
lower bounds can lead to new and interesting online BST algorithms.

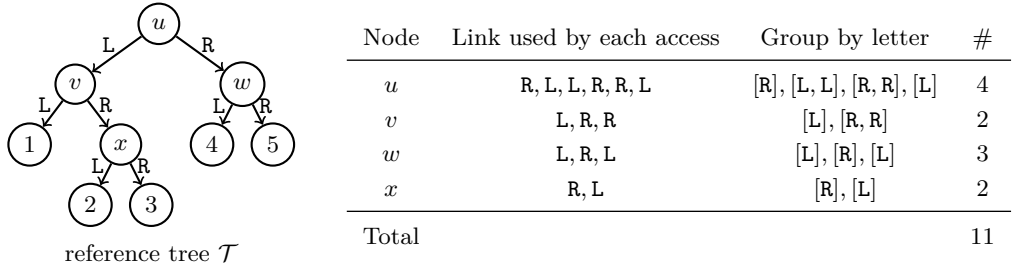
\begin{figure}[h]
\centering
\begin{tabu}{c}
\begin{tikzpicture}
\node[circle,draw] (u) at (3.125,4) {$u$};
\node[circle,draw] (v) at (1.75,3) {$v$};
\node[circle,draw] (w) at (4.5,3) {$w$};
\node[circle,draw] (x) at (2.5,2) {$x$};
\node[circle,draw] (k1) at (1,2) {$1$};
\node[circle,draw] (k2) at (2,1) {$2$};
\node[circle,draw] (k3) at (3,1) {$3$};
\node[circle,draw] (k4) at (4,2) {$4$};
\node[circle,draw] (k5) at (5,2) {$5$};
\newcommand{\arrowL}[2]{\draw[->] (#1) -- (#2) node[pos=.5, above left=-1mm and -1mm] {$\ttL{}$};}
\newcommand{\arrowR}[2]{\draw[->] (#1) -- (#2) node[pos=.5, above right=-1mm and -1mm] {$\ttR{}$};}
\newcommand{\children}[3]{\arrowL{#1}{#2}\arrowR{#1}{#3}}
\children{u}{v}{w}
\children{v}{k1}{x}
\children{w}{k4}{k5}
\children{x}{k2}{k3}
\end{tikzpicture}\\
reference tree $\cT$\\
\end{tabu}
\quad
\begin{tabu}{cccc}
Node & Link used by each access & Group by letter& $\#$\\
\midrule
$u$ & $\mathtt{R,L,L,R,R,L}$ & $\mathtt{[R],[L,L],[R,R],[L]}$ & 4\\%
$v$ & $\mathtt{L,R,R}$ & $\mathtt{[L],[R,R]}$ & 2\\
$w$ & $\mathtt{L,R,L}$ & $\mathtt{[L],[R],[L]}$ & 3\\
$x$ & $\mathtt{R,L}$ & $\mathtt{[R],[L]}$ & 2\\
\midrule
Total & & & 11\\
\end{tabu}

\caption{For access sequence $X = (4,1,3,5,4,2)$ and reference tree $\cT$, $\ALT_\cT(X) = 11$.}\label{fig:alt}
\end{figure}

\paragraph*{The \FB{}} The definition of Wilber's second bound, the \emph{\FB{}}, is less intuitive (and as such, 
was much less understood prior to this work).  
Let $G_\x$ be the set of $m$ points in the plane given by the map $x_i \mapsto (x_i,i)$. 
The \emph{funnel} of a point $p \in G_\x$ is the set of ``orthogonally visible'' 
points below $p$, i.e. points $q$ such that the axis-aligned rectangle with corners at $p$ and $q$ contains no other points (see Figure~\ref{fig:fb}).
For each $p$, look at the points in the funnel of $p$ sorted by $y$ coordinate, and count the number 
of alternations from the left to the right of $P$ that occur. 
Call this $f(p)$; this is $p$'s contribution to the lower bound.
Summing this value for all $p \in G_\x$  gives the lower bound $\W(\x) := \sum_{p \in G_\x} f(p)$. 
An algorithmic view of this bound is as follows:
consider the algorithm that simply brings each $x_i$ to the root by a series of single rotations.
Then $f(p)$ for $p=(x_i,i)$
is exactly the number of \emph{turns} on the path from the root to $x_i$ right before it is accessed \cite{AM78, IaconoSurvey}.
This view emphasizes the amortized nature of the funnel bound: at any point, there could be 
linearly many keys in the tree that are only \emph{one} turn away from the root, so one can only 
hope to achieve this bound in some amortized fashion. This partially explains why Wilber's second
bound has been so elusive to analyze (more on this interpretation can be found in the recent work of \cite{LT19}). 

\begin{figure}[h]
\centering
\newcommand{\diskColor}{gray}
\newcommand{\shadowColor}{lightgray}
\newcommand{\bottom}{0}
\newcommand{\funnel}[3]{
    \fill[gray] (#1,#2) circle (3mm);
    \node at (#1,#2) (C) {};
    \node[above=0mm of C] {#3};
}
\newcommand{\funnelL}[2]{
    \fill[lightgray] (#1,#2) rectangle (-0.5,\bottom);
    \funnel{#1}{#2}{$\ttL$}
}
\newcommand{\funnelR}[2]{
    \fill[lightgray] (#1,#2) rectangle (8.5,\bottom);
    \funnel{#1}{#2}{$\ttR$}
}

\newcommand{\access}[2]{\node[cross] (p#2) at (#1,#2) {};}
\begin{tabu}{c}
\begin{tikzpicture}[scale=0.5]
\funnelL{3}{3}
\funnelL{2}{7}
\funnelL{1}{8}
\funnelR{5}{4}
\funnelR{7}{6}
\fill[\shadowColor] (-0.5,\bottom) rectangle (8.5,1);
\access{4}{1}
\access{6}{2}
\access{3}{3}
\access{5}{4}
\access{1}{5}
\access{7}{6}
\access{2}{7}
\access{1}{8}
\access{4}{9}
\access{6}{10}
\access{3}{11}
\node[above=0mm of p9] {$p$};
\end{tikzpicture}\\
the funnel of $p$ has 5 points (highlighted)
\end{tabu}
\quad
\begin{tabu}{l}
Sorted by increasing $y$-coordinate: $\mathtt{L,R,R,L,L}$.\\
This forms 3 groups $\mathtt{[L],[R,R],[L,L]}$, so $f(p) = 3$.
\end{tabu}

\caption{Computing $f(p)$ for $p=(4,9)$ in the geometric view of $X=(4,6,3,5,1,7,2,4,6,3)$. Notice how the funnel points form a staircase-like front on either side of $p$.}\label{fig:fb}
\end{figure}

Wilber conjectured that
$\W(\x)\geq \Omega(\ALT(\x))$ for every access sequence $\x$, and that 
the \FB{} is in fact \emph{dynamically optimal}, i.e., that $\W(\x) = \Theta(\OPT(\x)) \; \forall \x$.
These conjectures were echoed multiple times in the long line of research spanning dynamic optimality 
(see e.g., \cite{DHIKP09,IaconoSurvey,CGKMS15,KS18}). Very recently, Levy and Tarjan 
\cite{LT19} gave a compelling intuitive explanation for why $\W(\x)$ is related to the amortized analysis 
of splay trees (see Section 4). Despite all this, the \FB{} remained elusive and no progress 
was made on Wilber's conjectures for nearly 40 years (To the best of our knowledge, the only properties 
that were previously known about the Funnel bound is that it is optimal in the ``key-independent'' setting \cite{Iacono05} 
and ``approximately monotone'' \cite{LT19}, both are prerequisites for dynamic optimality.)

Our main contribution affirmatively answers Wilber's first question, and settles the 
relationship between the \IB{} and the \FB{}:
\begin{theorem}[$\W$ dominates $\ALT$] \label{thm:domination}
For every access sequence $\x$ without repeats\footnote{As explained at the beginning of 
Section~\ref{sec:prelims}, it is fine for our purposes to focus on access sequences where each value appears only once.} 
and for every tree $\cT$, $\ALT_\cT(\x) \leq O(\W(\x) + m)$.
\end{theorem}
\begin{theorem}[Tight separation]\label{thm:separation}
There is an access sequence $\xTilde$ for which $\W(\xTilde) \geq \Omega(\lg\lg n) \cdot (\ALT_\cT(\xTilde)+m)$ 
simultaneously for all trees $\cT$.
\end{theorem}

The latter separation is tight up to constant factors, since Tango trees imply that $\OPT(\x) \leq O(\lg\lg n)\cdot \ALT(\x)$. 
An interesting corollary of Theorem \ref{thm:separation} is that the analysis of Tango trees cannot be improved by 
choosing \emph{any} reference tree, answering an open question of Iacono \cite{IaconoSurvey}. (One 
attractive idea is to choose a \emph{random} reference tree instead of the canonical balanced BST, but Theorem 
\ref{thm:separation} shows that this will not help in general.)

\paragraph*{A symmetric characterization of the Funnel bound} 

The geometric equivalence of dynamic optimality (through ``arborally satisfied'' rectangles \cite{DHIKP09}) makes it clear 
that $\OPT(X)$ is 
\emph{invariant} under geometric transformations of the access sequence $X$. Indeed, a fundamental barrier in understanding the 
\FB{} and its claim to optimality is that
it was unclear whether Wilber's  
bounds were invariant under \emph{rotations} of the access sequence $X$. Demaine et al. explicitly pointed out this challenge:

\begin{quote}
``It is also unclear how [Wilber's] bounds are affected by 90-degree rotations of the point set representing the access sequence and, for 
the Funnel bound, by \emph{flips}. Computer search reveals many examples where the bounds change slightly, and proving that they change 
by only a constant factor seems daunting.'' \cite{DHIKP09}
\end{quote}
 
This shows that \emph{exact} symmetry of $\W(\x)$ is hopeless, and  can only hold in some  `amortized' sense. 
Indeed, the heart of our paper, which is also a key ingredient in the proof of Theorem \ref{thm:domination}, 
is a new \emph{symmetric} characterization of the Funnel bound, which proves that, up to a $\pm O(m)$ additive term, 
it is indeed invariant to rotations. More formally, we show that for any access sequence $\x$, $\W(\x)$ is asymptotically 
equal to the number of occurrences in $G_X$ of a configuration of 4 points that we call 
a \textbf{\zr{}}\footnote{We thank an anonymous reviewer for informing us that \zr{}s have been discussed in the past under the name ``pinwheel configuration'', though (to the best of their knowledge) never in writing.} (see Figure~\ref{fig:zrect}).

\begin{figure}[h]
\centering

\newcommand{\access}[2]{\node[cross] (p#2) at (#1,#2) {};}
\begin{tabu}{c}
\begin{tikzpicture}[scale=0.5]
\draw[gray] (1,1) rectangle (4,4);
\access{3}{1}
\access{1}{2}
\access{4}{3}
\access{2}{4}
\end{tikzpicture}\\
\cmark\\
\zr{}\\
\end{tabu}
\quad
\begin{tabu}{c}
\begin{tikzpicture}[scale=0.5]
\draw[gray] (1,1) rectangle (4,4);
\access{3}{1}
\access{1}{2}
\access{4}{3}
\access{2}{4}
\access{2.5}{2.5}
\end{tikzpicture}\\
\xmark\\
wrong\\
\end{tabu}
\quad
\begin{tabu}{c}
\begin{tikzpicture}[scale=0.5]
\draw[gray] (1,1) rectangle (4,4);
\access{2}{1}
\access{1}{2}
\access{4}{3}
\access{3}{4}
\end{tikzpicture}\\
\xmark\\
wrong\\
\end{tabu}
\caption{A \zr{} is a configuration of 4 points. Its interior must be empty, and the relative order of the four points matters.}\label{fig:zrect}
\end{figure}
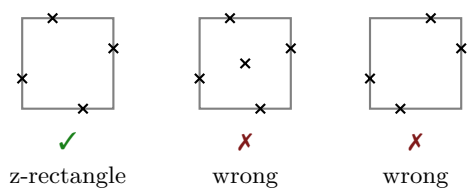

A crucial difference between \zr{}s and the notion of \emph{independent rectangles}  \cite{DHIKP09}
is that the latter have to satisfy additional \emph{independence} constaints across several rectangles,
whereas \zr{}s have no ``global'' constraints whatsoever. 
In other words, \zr{}s are a \emph{local} feature of the access sequence, in the
sense that their existence and contribution to the lower bound are unaffected by other \zr{}s and by points outside of it. 
We believe this key property will make the analysis of online BST algorithms
more tractable, as it gives a simpler competitive benchmark. We next describe an initial step in this direction. 

\paragraph*{Towards dynamic optimality of the Funnel Bound}
One consequence of the simplicity of the \zr{} characterization of the Funnel bound is that it makes it easier to compare it 
both to other BST lower bounds and to candidate algorithms for dynamic optimality.  As a proof of concept, we show that when 
there is only a constant number of \zr{} in $G_X$, then
$\upIRB(X)$ is linear, where $\upIRB$ is one of the terms in the \emph{Independent Rectangle} bound
$\IRB(X) := \upIRB(X) + \downIRB(X)$, which is known to dominate both of Wilber's bounds \cite{DHIKP09} (we define $\upIRB(X)$
\ifarxiv
    in Section~\ref{sec:speculation}).
\fi
More formally, 
\begin{theorem}\label{thm:speculation}
If $G_X$ contains $O(1)$ \zr{}s, then $\upIRB(X) \leq O(m)$.
\end{theorem}

We remark that the proof of this theorem already introduces a nontrivial 
charging argument that could (hopefully) be generalized to prove that $\W$ matches $\IRB$, as conjectured by previous works \cite{IaconoSurvey}. 

\paragraph*{Techniques} At a very high level, 
the main ideas in Theorem~\ref{thm:domination} are to use the self-reducible structure of the \IB{}, and to show that interleaving two access 
sequences $\xL$ and $\xR$ on \emph{disjoint ranges} is a super-additive operation, i.e., it increases the overall value of $\Funnel(X)$ to more than the sum of its parts $\Funnel(\xL)+\Funnel(\xR)$. This argument involves both $X$ and its \emph{reverse} (flip), hence our new symmetric characterization of the \FB{} (through z-rectangles) is key to the proof. The main idea behind Theorem~\ref{thm:separation} is to form hard sequences over geometrically-spaced sets of keys $\{i+1, i+2, i+4, i+8, \ldots\}$, each of which can ``force'' $\ALT_\cT$ to pick a very lopsided reference tree $\cT$. Those sequences can then be concatenated together so that the average value of $\ALT_\cT$ is provably low whichever $\cT$ was picked. Finally, the key idea in Theorem~\ref{thm:speculation} is to study the consequences of the \emph{absence} of \zr{}s on the combinatorial structure of point set $G_X$, and use this to bound the value of $\upIRB(X)$ by a charging argument.

\paragraph*{Remark on independent work} In a concurrent and independent work, Chalermsook, Chuzhoy and Saranurak \cite{CCS19} obtain a (weaker) $\Theta(\lg \lg n/ \lg \lg \lg n)$ separation between $\ALT$ and $\W$, in the same spirit as the tight separation we give in Theorem~\ref{thm:separation}. Our works are otherwise unrelated.

\section{Preliminaries}\label{sec:prelims}
To make our definitions and proofs easier, we will work directly in the geometric representation of access sequences as (finite) sets of points in the plane $\R^2$.
\begin{definition}[geometric view]
Any access sequence $\x=(x_1, \ldots, x_m) \in [n]^m$ can be represented as the set of points $G_X=\{(x_i, i) \mid i \in [n]\}$, where the $x$-axis represents the key and the $y$-axis represents time (see Figure~\ref{fig:geom-view}).
\end{definition}

By construction, in $G_X$, no two points share the same $y$-coordinate. We will say such a set has ``distinct $y$-coordinates''. In addition, we note that it is fine to restrict our attention to sequences $X$ without repeated values.\footnote{Indeed, Appendix~E in \cite{CGKMS15} gives a simple operation that transforms any sequence $\x$ into a sequence $\spl(\x)$ without repeats such that $\OPT(\spl(\x)) = \Theta(\OPT(\x))$. Thus if we found a tight lower bound $L(X)$ for sequences without repeats, a tight lower bound for general $\x$ could be obtained as $L(\spl(\x))$.}
The geometric view $G_X$ of such sequences also has no two points with the same $x$-coordinate. We will say that such a set has ``distinct $x$- and $y$-coordinates''.

\begin{definition}[$x$- and $y$-coordinates]
For a point $p \in \R^2$, we will denote its $x$- and $y$-coordinates as $p.x$ and $p.y$. Similarly, we define $P.x = \{p.x \mid p \in P\}$ and $P.y = \{p.y \mid p \in P\}$.
\end{definition}

We start by defining the \emph{mixing value} of two sets: a notion of how much two sets of numbers are interleaved. It will be useful in defining both the \IB{} and the \FB{}. We define it in a few steps.
\begin{definition}[mixing string]
Given two disjoint finite sets of real numbers $L,R$, let $\mix(L,R)$ be the string in $\{\ttL,\ttR\}^*$ that is obtained by taking the union $L \cup R$ in increasing order and replacing each element from $L$ by $\ttL$ and each element from $R$ by $\ttR$. For example, $\mix(\{2,3,8\},\{1,5\}) = \mathtt{RLLRL}$.
\end{definition}

\begin{definition}[number of blocks]
Given a string $s \in \{\ttL,\ttR\}^*$, we define $\blocks(s)$ as the number of contiguous blocks of the same symbol in $s$. Formally,
\[\blocks(s) \coloneqq 
\begin{cases}
0\text{ if $s$ is empty}\\
1 + \#\{i \mid s_i \neq s_{i+1}\}\text{ otherwise.}
\end{cases}\]
For example,
$\blocks(\mathtt{LLLRLL})=3$. Note that if we insert characters into $s$, $\blocks(s)$ can only increase.
\end{definition}

\begin{definition}[mixing value]
Let $\mixValue(L,R)\coloneqq \blocks(\mix(L,R))$ (see Figure~\ref{fig:mixvalue}).  
\end{definition}

\begin{figure}[h]
\centering
\newcommand{\pointL}[1]{
    \node[dot] (p#1) at (#1,2) {};
    \node[above=1.5mm of p#1] {$#1$};
}
\newcommand{\pointR}[1]{
    \node[dot] (p#1) at (#1,0) {};
    \node[below=1.5mm of p#1] {$#1$};
}
\newcommand{\surround}[3]{
    \pgfmathsetmacro{\rad}{.35}
    \pgfmathsetmacro{\bottom}{#1-\rad}
    \pgfmathsetmacro{\top}{#1+\rad}
    \fill[rounded corners, lightgray] (#2-\rad,\bottom) rectangle (#3+\rad,\top) {};
}
\newcommand{\surroundL}[2]{\surround{2}{#1}{#2}}
\newcommand{\surroundR}[2]{\surround{0}{#1}{#2}}
\begin{tikzpicture}[scale=0.5]
\surroundL{1}{3}
\surroundL{6}{6}
\pointL{1}
\pointL{3}
\pointL{6}
\surroundR{4}{4}
\surroundR{7}{8}
\pointR{4}
\pointR{7}
\pointR{8}
\node at (-0.8,2) {$L$};
\node at (-0.8,0) {$R$};
\end{tikzpicture}
\caption{A visualization of $\mixValue(\{1,3,6\},\{4,7,8\}) = 4$}\label{fig:mixvalue}
\end{figure}
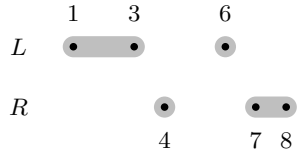

The mixing value has some convenient properties, which we will use later:
\begin{fact}[properties of $\mixValue$]\label{fact:props-mixvalue}
Function $\mixValue(L,R)$ is:
\begin{alphaenumerate}
\item symmetric: $\mixValue(L,R) = \mixValue(R,L)$;
\item monotone: if $L_1 \subseteq L_2$ and $R_1 \subseteq R_2$, then $\mixValue(L_1,R_1) \leq \mixValue(L_2,R_2)$;
\item subadditive under concatenation: if $L_1,R_1 \subseteq (-\infty,x]$ and $L_2,R_2 \subseteq [x,+\infty)$, then $\mixValue(L_1 \cup L_2, R_1 \cup R_2) \leq \mixValue(L_1,R_1) + \mixValue(L_2,R_2)$.
\end{alphaenumerate}
Finally, $\mixValue(L,R) \leq 2\cdot\min(|L|,|R|)+1$.
\end{fact}

We now give precise definitions of Wilber's two bounds.\footnote{These definitions may differ by a constant factor or an additive $\pm O(m)$ from the definitions the reader has seen before.
We will ignore such differences, because the cost of a BST also varies by $\pm O(m)$ depending on the definition, and the interesting regime is when $\OPT(X) = \omega(m)$.}

\begin{definition}[\IB{}]\label{def:ib}
Let $P$ be a point set with distinct $y$-coordinates, and let $\cT$ be a binary tree in which leaves are labeled with
elements of $P.x$
in increasing order, and each non-leaf node has two children.

We define $\ALT_\cT(P)$ using the recursive structure of $\cT$. If $\cT$ is a single node, let $\ALT_\cT(P)\coloneqq 0$. Otherwise, let $\TL$ and $\TR$ be the left and right subtrees at the root. Partition $P$ into two sets $\PL\coloneqq \{p \in P \mid p.x \in \TL\}$ and $\PR \coloneqq \{p \in P \mid p.x \in \TR\}$. Define quantity
\[a(P,\cT) \coloneqq \mixValue(\PL.y, \PR.y),\]
which describes how much $\PL$ and $\PR$ are interleaved in time. Then
\begin{equation}\label{eq:ib}
\ALT_\cT(P) \coloneqq a(P,\cT) + \ALT_\TL(\PL) + \ALT_\TR(\PR).
\end{equation}
In addition, for an access sequence $X$, let $\ALT_\cT(X) \coloneqq \ALT_\cT(G_X)$.
\end{definition}

\begin{definition}[axis-aligned rectangle delimited two points]
Given two points $p$ and $q$ with distinct $x$- and $y$- coordinates, let $\rect{pq}$ be the smallest axis-aligned rectangle that contains both $p$ and $q$. Formally,
\[\rect{pq} \coloneqq [\min(p.x, q.x), \max(p.x, q.x)] \times [\min(p.y, q.y), \max(p.y, q.y)].\]
\end{definition}

\begin{definition}[empty rectangles]
Let $P$ be a point set. Given $p,q \in P$, we say $\rect{pq}$ is 
empty\footnote{This corresponds to the notion of ``unsatisfied rectangle'' in \cite{DHIKP09}.} in $P$ if $P \cap \rect{pq} = \{p,q\}$ (see Figure~\ref{fig:rects}).
\end{definition}

\begin{figure}[h]
\centering

\newcommand{\access}[2]{\node[cross] (p#2) at (#1,#2) {};}
\setlength{\tabcolsep}{5mm}
\begin{tabu}{cc}
\begin{tikzpicture}[scale=0.5]
\draw[gray] (1,1) rectangle (3,4);
\access{1}{1}
\node[left=0mm of p1] {$p$};
\access{3}{4}
\node[right=0mm of p4] {$q$};
\access{0}{3}
\end{tikzpicture}&
\begin{tikzpicture}[scale=0.5]
\draw[gray] (1,1) rectangle (4,3);
\access{1}{3}
\node[left=0mm of p3] {$r$};
\access{4}{1}
\node[right=0mm of p1] {$s$};
\access{2}{2}
\end{tikzpicture}\\
$\rect{pq}$ is empty&
$\rect{rs}$ is \emph{not} empty\\
\end{tabu}

\caption{Some axis-aligned rectangles}\label{fig:rects}
\end{figure}
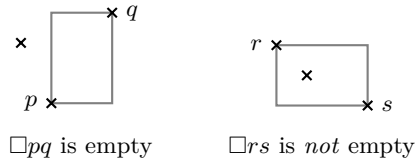

For the next definitions, it is helpful to refer back to Figure~\ref{fig:fb}. In particular, $\FL$ and $\FR$ (the left and right funnel) correspond to the points marked with $\ttL$ and $\ttR$.

\begin{definition}[left and right funnel]\label{def:lr-funnel}
Let $P$ be a point set. For each $p \in P$, we say that access $q \in P$ is in the left (resp. right) funnel of $p$ within $P$ if $q$ is to the lower left (resp. lower right) of $p$ and $\rect{pq}$ is empty. Formally, let
\[\FL \coloneqq \{q \in P \mid q.y < p.y \,\wedge\, q.x < p.x \,\wedge\, P \cap \rect{pq} = \{p,q\}\}\]
and
\[\FR \coloneqq \{q \in P \mid q.y < p.y \,\wedge\, q.x > p.x \,\wedge\, P \cap \rect{pq} = \{p,q\}\}.\]
We will collectively call $\FL \cup \FR$ the funnel of $p$ within $P$.
\end{definition}

\begin{definition}[\FB{}]\label{def:fb}
Let $P$ be a point set with distinct $y$-coordinates. For each $p \in P$, define quantity
\[f(P,p) \coloneqq \mixValue(\FL.y, \FR.y),\]
which describes how much the left and right funnel of $p$ are interleaved in time.
Then
\[\W(P) \coloneqq \sum_{p \in P} f(P,p).\]
In addition, for an access sequence $X$, let $\W(X) \coloneqq \W(G_X)$.
\end{definition}

\section{The \FB{} dominates the \IB{}}
We prove that $\Funnel$ dominates $\ALT$ in two parts: in Section~\ref{sec:domination1} we show that $\ALT(X)$ is dominated by the sum $\Funnel(X) + \Funnel(\rev{X})$, where $\rev{X}$ is the reverse of $X$, then in Section~\ref{sec:zrects} we prove that $\Funnel(\rev{X}) \approx \Funnel(X)$ using our new characterization of $\Funnel$ by \zr{}s.

\subsection{Upper-bounding the \IB{} by a sum of two \FB{}s}\label{sec:domination1}

\begin{definition}[time reversal]
The time reversal of a point $p \in \R^2$ is $\rev{p} \coloneqq (p.x, -p.y)$.\footnote{The notation is inspired from the notion of complex conjugate, which is also a vertical flip.} The time reversal of a point set $P$ is $\rev{P}\coloneqq \{\rev{p} \mid p \in P\}$ (see Figure~\ref{fig:time-reversal}).
\end{definition}

\begin{figure}[h]
\centering
\newcommand{\access}[2]{
\node[cross] (p#2) at (#1,#2) {};
}
\setlength{\tabcolsep}{10mm}
\begin{tabu}{cc}
\begin{tikzpicture}[scale=0.4]
\access{0}{0}
\access{4}{1}
\node[right=0mm of p1] {$p$};
\access{1}{2}
\access{3}{3}
\access{6}{4}
\access{5}{5}
\access{2}{6}
\end{tikzpicture}&
\begin{tikzpicture}[scale=0.4]
\access{0}{7}
\access{4}{6}
\node[right=0mm of p6] {$\rev{p}$};
\access{1}{5}
\access{3}{4}
\access{6}{3}
\access{5}{2}
\access{2}{1}
\end{tikzpicture}\\[1mm]
$P$ & $\rev{P}$\\
\end{tabu}
\caption{A point set and its time reversal}\label{fig:time-reversal}
\end{figure}
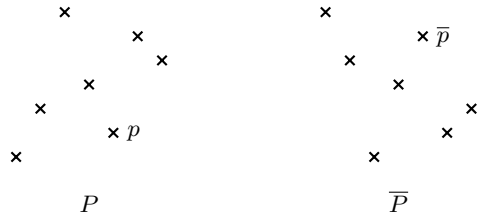

We first prove the following lemma.

\begin{lemma}\label{lemma:two-sided}
Let $P$ be a point set with distinct $y$-coordinates, and let $\cT$ be a tree that satisfies the conditions of Definition~\ref{def:ib}. Then $\W(P) + \W(\rev{P}) \geq \ALT_\cT(P)$.
\end{lemma}

Even though the formal proof of this lemma is a relatively involved case analysis, it is easy to understand geometrically. The key observation is the following. Consider two sequences $\xL$ and $\xR$ on disjoint ranges, and interleave to form a single sequence $\x$. Then the more times we switch from elements of $\xL$ to elements of $\xR$, the bigger $\Funnel(\x)+\Funnel(\rev{\x})$ is going to be.

To see this, let's look at the geometric view of $\x$ (see Figure~\ref{fig:domination-idea}). Let $p$ and $q$ be two consecutive points on the $\xL$ side that are separated by a streak of points from $\xR$ (i.e. all accesses between $p$ and $q$ vertically are from $\xR$). First, assume $p.x > q.x$. Then $q$ is in the left funnel of $p$, and at least of the points on the $\xR$ between $p$ and $q$ must be in the right funnel of $p$, which forms a completely new group of funnel points compared to what $p$ had in $\xL$. This means that the contribution of $p$ to $\Funnel(\x)$ is at least one higher than its contribution to $\Funnel(\xL)$.

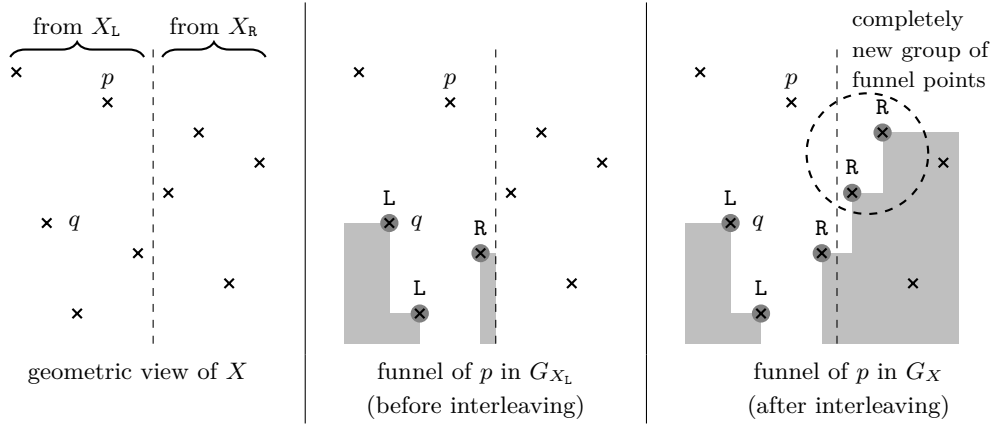
\begin{figure}[h]
\centering
\newcommand{\diskColor}{gray}
\newcommand{\shadowColor}{lightgray}
\newcommand{\bottom}{-2}
\newcommand{\funnel}[3]{
    \fill[gray] (#1,#2) circle (3mm);
    \node at (#1,#2) (C) {};
    \node[above=0mm of C] {#3};
}
\newcommand{\funnelL}[3]{
    \fill[lightgray] (#1,#2) rectangle (#3,\bottom);
    \funnel{#1}{#2}{$\ttL$}
}
\newcommand{\funnelR}[3]{
    \fill[lightgray] (#1,#2) rectangle (#3,\bottom);
    \funnel{#1}{#2}{$\ttR$}
}
\newcommand{\access}[2]{\node[cross] (p#2) at (#1,#2) {};}
\newcommand{\pointSet}{
\access{3}{-1}
\access{8}{0}
\access{5}{1}
\access{2}{2}
\access{6}{3}
\access{9}{4}
\access{7}{5}
\access{4}{6}
\access{1}{7}
\node[above=0mm of p6] {$p$};
\node[right=1mm of p2] {$q$};
\draw[dashed, thin] (5.5,\bottom) -- (5.5,8);
}

\setlength{\tabcolsep}{5mm}
\begin{tabu}{c|c|c}
\begin{tikzpicture}[scale=0.5]
\pointSet{}
\draw [decorate,decoration={brace,amplitude=2mm},yshift=5mm]
(0.8,7) -- (5.2,7) node [midway,yshift=4mm] {from $\xL$};
\draw [decorate,decoration={brace,amplitude=2mm},yshift=5mm]
(5.8,7) -- (9.2,7) node [midway,yshift=4mm] {from $\xR$};
\end{tikzpicture}&
\begin{tikzpicture}[scale=0.5]
\funnelL{3}{-1}{0.5}
\funnelL{2}{2}{0.5}
\funnelR{5}{1}{5.5}
\pointSet{}
\end{tikzpicture}&
\begin{tikzpicture}[scale=0.5]
\funnelL{3}{-1}{0.5}
\funnelL{2}{2}{0.5}
\funnelR{5}{1}{9.5}
\funnelR{6}{3}{9.5}
\funnelR{7}{5}{9.5}
\pointSet{}
\draw[dashed] (6.5,4.3) circle (2cm);
\node[text width=2cm,align=left] at (8.5,7.7) {\footnotesize completely new group of funnel points};
\end{tikzpicture}\\
geometric view of $\x$&
funnel of $p$ in $G_{\xL}$&
funnel of $p$ in $G_{\x}$\\
&(before interleaving)&(after interleaving)\\
\end{tabu}

\caption{Interleaving sequences $\xL=(3,5,2,4,1)$ and $\xR=(8,6,9,7)$ into $X=(3,8,5,2,6,9,7,4,1)$. The contribution of $p$ to $\Funnel(\xL)$ is 3, while the contribution of $p$ to $\Funnel(\x)$ is 4.}\label{fig:domination-idea}
\end{figure}

What if $p.x < q.x$ instead? Then it turns out that an analogous argument can be made on $q$ if we take the time reversal of $X$. That is, the contribution of $\rev{q}$ to $\Funnel(\rev{\x})$ is at least one higher than its contribution to $\Funnel(\rev{\xL})$. Indeed, if we flip the point set vertically, then $p$ and $q$ exchange roles, which means $p.x > q.x$ once again.

To conclude, it remains to observe that the $a(P,p)$ term in the recursive definition of $\ALT_\cT(\x)$ is precisely a measure of how much the subsequences $\xL$ and $\xR$ corresponding to the left and right subtree at the root of $\cT$ are interleaved. So we can apply the argument above by induction to show that $\Funnel(\x) + \Funnel(\rev{\x}) \geq \ALT_\cT(\x)$.
We now reluctantly move to the formal proof.

\begin{proof}[Proof of Lemma~\ref{lemma:two-sided}]
We prove this by induction on $\cT$. The base case is $\cT$ made of a single node. In this case, $\ALT_\cT(P)=0$ by definition, so the inequality trivially holds.

Now consider a general tree $\cT$, and define $\TL$, $\TR$, $\PL$ and $\PR$ as in Definition~\ref{def:ib}. Note that each leaf of $\cT$ has a label in $P.x$ and $\TL$ and $\TR$ must each have at least one leaf, so $\PL$ and $\PR$ are not empty. 
Let's apply the induction hypothesis on $(\PL,\TL)$ and $(\PR,\TR)$. This means that
\begin{align*}
\W(\PL) + \W(\rev{\PL}) &\geq \ALT_\TL(\PL)\\
\W(\PR) + \W(\rev{\PR}) &\geq \ALT_\TR(\PR).
\end{align*}
Thus we find that
\begin{align}
\ALT_\cT(P) &= a(P,\cT) + \ALT_\TL(\PL) + \ALT_\TR(\PR)\tag{by definition}\\
&\leq a(P,\cT) + \W(\PL) + \W(\rev{\PL}) + \W(\PR) + \W(\rev{\PR})\label{eq:applied-hyp}
\end{align}

\begin{claim}
If $p \in \PL$, then
\[f(P,p) \geq f(\PL,p)\quad\text{and}\quad f(\rev{P},\rev{p}) \geq f(\rev{\PL},\rev{p});\]
and if $p \in \PR$, then
\[f(P,p) \geq f(\PR,p)\quad\text{and}\quad f(\rev{P},\rev{p}) \geq f(\rev{\PR},\rev{p}).\]
\end{claim}

\begin{claimproof}
We will deal with the first case (the other three cases are symmetric). The key is that $\PL$ and $\PR$ operate on disjoint ranges of $x$-coordinates.
\begin{itemize}
\item The left funnel of $p$ within $\PL$ is identical to its left funnel within $P$, since all elements of $\PR$ are to the right of $p$. Formally, $F_\ttL(\PL,p) = \FL$.
\item All points $q$ that were in the right funnel of $p$ within $\PL$ will still be part of the right funnel of $p$ within $P$. Indeed, the only way for them to stop being funnel points would be to add accesses inside the rectangle delimited by $p$ and $q$. This doesn't happen because all points in $\PR$ are strictly to the right of all points in $\PL$. Formally, $F_\ttR(\PL,p) \subseteq \FR$.
\end{itemize}
Therefore, $\mix(F_\ttL(\PL,p).y,F_\ttR(\PL,p).y)$ is a subsequence of $\mix(\FL.y,\FR.y)$, which means that
\[f(\PL,p) = \blocks(\mix(F_\ttL(\PL,p),F_\ttR(\PL,p))) \leq \blocks(\mix(\FL,\FR)) = f(P,p).\]
\end{claimproof}

Summing up $f(P,p)$ and $f(\rev{P},\rev{p})$ over all points $p \in P$, we obtain
\begin{equation}\label{eq:sum-up}
\begin{aligned}
\W(P) &= \sum_{p \in P} f(P,p) \geq \sum_{p \in \PL} f(\PL,p) + \sum_{p \in \PR} f(\PR,p) = \W(\PL) + \W(\PR)\\
\W(\rev{P}) &= \sum_{p \in P} f(\rev{P},\rev{p}) \geq \sum_{p \in \PL} f(\rev{\PL},\rev{p}) + \sum_{p \in \PR} f(\rev{\PR},\rev{p}) = \W(\rev{\PL}) + \W(\rev{\PR}).
\end{aligned}
\end{equation}
This, combined with \eqref{eq:applied-hyp}, gives
\begin{align*}
\W(P) + \W(\rev{P})
&\geq \W(\PL) + \W(\PR) + \W(\rev{\PL}) + \W(\rev{\PR})\\
&\geq \ALT_\cT(P) - a(P,\cT)
\end{align*}
This falls $a(P,\cT)$ short of our goal (which makes sense, since we haven't used the interleaving of $\PL$ and $\PR$ yet). To fix this, we will show the following claim.
\begin{claim}\label{claim:scenarios}
Consider the following properties defined over a point $p \in P$:
\begin{alphaenumerate}
\item $p \in \TL$ and $f(P,p) \geq f(\PL,p)+1$;
\item $p \in \TL$ and $f(\rev{P},\rev{p}) \geq f(\rev{\PL},\rev{p})+1$;
\item $p \in \TR$ and $f(P,p) \geq f(\PR,p)+1$;
\item $p \in \TR$ and $f(\rev{P},\rev{p}) \geq f(\rev{\PR},\rev{p})+1$.
\end{alphaenumerate}
The sum of the number of points in $P$ having each property (a)--(d) is at least $a(P,\cT)$.
\end{claim}

\begin{claimproof}
Let's number the points of $P$ by increasing $y$-coordinate (i.e. in chronological order) as $p_1, \ldots, p_m$.
Recall that $a(P,\cT) = \mixValue(\PL.y, \PR.y)$. Also, $\PL$ and $\PR$ are non-empty, so $a(P,\cT) \geq 2$. This means that as we go through the points $p_1, \ldots, p_m$, we switch $a(P,\cT)-1 \geq 1$ times between points of $\PL$ and points of $\PR$.

Therefore, there are exactly $a(P,\cT)-2$ pairs of indices $(i,j)$ with $i+1<j$ such that
\begin{itemize}
\item case 1: $p_i,p_j \in \PL$ but $p_{i+1}, \ldots, p_{j-1} \in \PR$, or
\item case 2: $p_i,p_j \in \PR$ but $p_{i+1}, \ldots, p_{j-1} \in \PL$,
\end{itemize}
which ``straddle accesses of the opposite side''.
Also, there is an index $i^*>1$ (the ``first element of the side that starts appearing later'') such that
\begin{itemize}
\item case 3: $p_{i^*} \in \PL$ but $p_1, \ldots, p_{i^*-1} \in \PR$, or
\item case 4: $p_{i^*} \in \PR$ but $p_1, \ldots, p_{i^*-1} \in \PL$
\end{itemize}
and similarly, there is an index $j^*<m$ (the ``last element of the side that finishes appearing earlier'') such that
\begin{itemize}
\item case 5: $p_{j^*} \in \PL$ but $p_{j^*+1}, \ldots, p_m \in \PR$, or
\item case 6: $p_{j^*} \in \PR$ but $p_{j^*+1}, \ldots, p_m \in \PL$.
\end{itemize}
This makes for a total of $a(P,\cT)-2+1+1=a(P,\cT)$ occurrences of one of the six cases. We will show that each of them leads to a point $p$ satisfying one of the properties (a)--(d). More precisely, we claim that:
\begin{itemize}
\item case 1 implies $p_j$ has property (a) or $p_i$ has property (b);
\item case 2 implies $p_j$ has property (c) or $p_i$ has property (d);
\item case 3 implies $p_{i^*}$ has property (a);
\item case 4 implies $p_{i^*}$ has property (c);
\item case 5 implies $p_{j^*}$ has property (b);
\item case 6 implies $p_{j^*}$ has property (d).
\end{itemize}

We will show this for case 1 and case 3. The other four cases are analogous.
To treat case 1, let's separate into more cases.\footnote{We wish we were joking.}
\begin{itemize}
\item If $p_i.x < p_j.x$, then $p_i$ is in the left funnel of $p_j$ within both $P$ and $\PL$. But within $P$, $p_{j-1}$ would be an additional right funnel point. Since it has a higher index than $p_i$, this would add at least 1 to $f(P,p_j)$ compared to $f(\PL, p_j)$. In other words, $f(P,p_j) \geq f(\PL,p_j)+1$ (scenario (a)).
\item If $p_i.x > p_j.x$, then we can use the same argument as above on $\rev{P}$ and $\rev{\PL}$ by swapping $i$ and $j$, obtaining $f(\rev{P},\rev{p_i}) \geq f(\PL,\rev{p_i})+1$ (scenario (b)).
\item If $p_i.x = p_j.x$, then both funnels of $p_j$ within $\PL$ are completely empty, which means that $f(\PL,p_j)=0$, while the right funnel of $p_j$ in $P$ would contain at least $p_{j-1}$. Therefore, $f(P,p_j) = 1 \geq f(\PL,p_j)+1$ (scenario (a)).
\end{itemize}
To treat case 3, it suffices to observe that both funnels of $p_{i^*}$ within $\PL$ would be completely empty (for lack of lower points), so $f(\PL, p_{i^*})=0$, while in $P$ the right funnel of $x_{i^*}$ would contain at least $p_{i^*-1}$. Therefore, $f(P,p_{i^*}) \geq 1 = f(\PL,p_{i^*}) + 1$ (scenario (a)).
\end{claimproof}

Now, if we sum up $f(P,p)$ and $f(\rev{P},\rev{p})$ over all points $p$ as we did in \eqref{eq:sum-up}, but this time also apply Claim~\ref{claim:scenarios}, we obtain that
\[\W(P) + \W(\rev{P}) \geq \W(\PL) + \W(\PR) + \W(\rev{\PL}) + \W(\rev{\PR}) + a(P,\cT).\]
Combined with \eqref{eq:applied-hyp}, this gives the desired result and concludes the inductive step.
\end{proof}

\subsection{Characterizing the \FB{} using \zr{}s} \label{sec:zrects}

Lemma~\ref{lemma:two-sided} asserts that all possible \IB{}s for all choices of reference trees $\cT$, are simultaneously upper-bounded 
by the sum of two specific \FB{}s. While this is already a nontrivial bound, $\W(P)$ and $\W(\rev{P})$ could in principle be wildly different,  
and it is therefore more compelling to show that the \emph{single} quantity $\W(P)$ already provides an upper bound. 
(It is curious that the symmetry properties of the Funnel bound, which are a necessary precondition for dynamic optimality, already 
enter the picture in determining the relationship between Wilber's bounds.)

To achieve this, we need to think about how geometric transformations affect the value of the \FB{}. It is clear from the definition that $\W(P)$ is unaffected by a horizontal flip. Indeed, the left funnel would become the right funnel and vice versa, so this wouldn't affect the number of times we switch between the two: the quantity $f(P,p)$ would remain the same for each $p$ (see Figure~\ref{fig:hor-flip}).

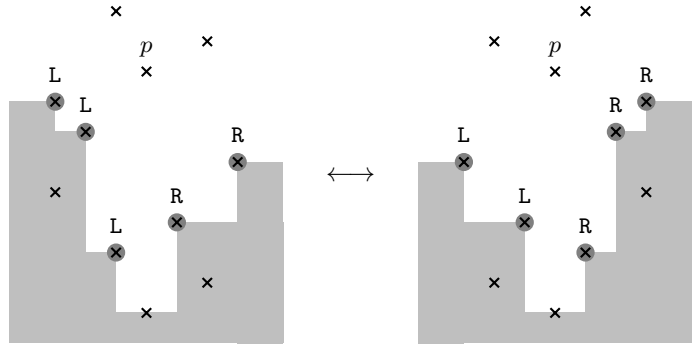
\begin{figure}[h]
\centering
\newcommand{\diskColor}{gray}
\newcommand{\shadowColor}{lightgray}
\newcommand{\bottom}{0}
\newcommand{\funnel}[3]{
    \fill[gray] (#1,#2) circle (3mm);
    \node at (#1,#2) (C) {};
    \node[above=0mm of C] {#3};
}
\newcommand{\funnelL}[2]{
    \fill[lightgray] (#1,#2) rectangle (-0.5,\bottom);
    \funnel{#1}{#2}{$\ttL$}
}
\newcommand{\funnelR}[2]{
    \fill[lightgray] (#1,#2) rectangle (8.5,\bottom);
    \funnel{#1}{#2}{$\ttR$}
}

\newcommand{\access}[2]{\node[cross] (p#2) at (#1,#2) {};}
$
\vcenter{\hbox{
\begin{tikzpicture}[scale=0.5]
\funnelL{3}{3}
\funnelL{2}{7}
\funnelL{1}{8}
\funnelR{5}{4}
\funnelR{7}{6}
\fill[\shadowColor] (-0.5,\bottom) rectangle (8.5,1);
\access{4}{1}
\access{6}{2}
\access{3}{3}
\access{5}{4}
\access{1}{5}
\access{7}{6}
\access{2}{7}
\access{1}{8}
\access{4}{9}
\access{6}{10}
\access{3}{11}
\node[above=0mm of p9] {$p$};
\end{tikzpicture}
}}
\quad
\longleftrightarrow
\quad
\vcenter{\hbox{
\begin{tikzpicture}[scale=0.5]
\funnelR{5}{3}
\funnelR{6}{7}
\funnelR{7}{8}
\funnelL{3}{4}
\funnelL{1}{6}
\fill[\shadowColor] (-0.5,\bottom) rectangle (8.5,1);
\access{4}{1}
\access{2}{2}
\access{5}{3}
\access{3}{4}
\access{7}{5}
\access{1}{6}
\access{6}{7}
\access{7}{8}
\access{4}{9}
\access{2}{10}
\access{5}{11}
\node[above=0mm of p9] {$p$};
\end{tikzpicture}
}}
$

\caption{Flipping the geometric view horizontally conserves the contribution $f(P,p)$ of each point: the only change is that the labels of the funnel points flip between $\ttL$ and $\ttR$.}\label{fig:hor-flip}
\end{figure}

On the other hand, it is far from obvious that the \FB{} is unaffected by a vertical flip. Because of the time reversal, the notion of funnel changes completely. And indeed, the precise value will change, as is shown in Figure~\ref{fig:ver-flip}.

\begin{figure}[h]
\centering
\newcommand{\access}[3]{
    \node[cross] (p#2) at (#1,#2) {};
    \node[above=0mm of p#2] {#3};
}

\setlength{\tabcolsep}{10mm}
\begin{tabu}{cc}
\begin{tikzpicture}[scale=0.5]
\access{1}{1}{$0$}
\access{3}{2}{$1$}
\access{2}{3}{$2$}
\end{tikzpicture}&
\begin{tikzpicture}[scale=0.5]
\access{2}{1}{$0$}
\access{3}{2}{$1$}
\access{1}{3}{$1$}
\end{tikzpicture}\\
$\Funnel(P) = 0+1+2 = 3$&
$\Funnel(\rev{P}) = 0+1+1 = 2$\\
\end{tabu}

\caption{A minimal example such that $\Funnel(P) \neq \Funnel(\protect\rev{P})$ is $P=\{(1,1),(3,2),(2,3)\}$. Each access $p$ is labeled with its contribution $f(P,p)$ (left) or $f(\protect\rev{P},p)$ (right).}\label{fig:ver-flip}
\end{figure}
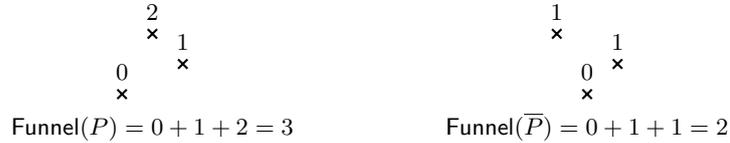

Nevertheless, we will show that for any point set $P$ with distinct $x$- and $y$-coordinates, $\W(P)$ and $\W(\rev{P})$ are equal up to an additive $O(m)$. We do this by introducing a new characterization of the \FB{} that is naturally invariant under \emph{90\textdegree{} rotations} of the point set. This new characterization is the number of \emph{\zr{}s}.

\begin{figure}[h]
\centering

\newcommand{\access}[2]{\node[cross] (p#2) at (#1,#2) {};}
\begin{tikzpicture}[scale=0.5]
\draw[gray] (2,2) rectangle (7,6);
\access{4}{1}
\access{5}{2}
\access{2}{3}
\access{7}{4}
\access{1}{5}
\access{3}{6}
\access{8}{7}
\access{6}{8}
\node[above=0mm of p6] {$p$};
\node[left=0mm of p3] {$q$};
\node[below=0mm of p2] {$r$};
\node[right=0mm of p4] {$s$};
\end{tikzpicture}
\caption{A \zr{}. The relative order of points $p,q,r,s$ horizontally and vertically matters.}\label{fig:zrect-formal}
\end{figure}
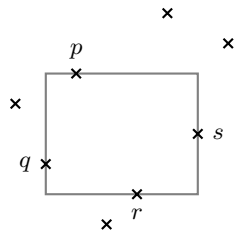

\begin{definition}[\zr]
Let $P$ be a point set. We call tuple $(p,q,r,s) \in P^4$ a \zr{} of $P$ if the following conditions hold:
\begin{alphaenumerate}
\item $q.x < p.x < r.x < s.x$;
\item $r.y < q.y < s.y < p.y$;
\item $P \cap [q.x,s.x] \times [r.y,p.y] = \{p,q,r,s\}$.
\end{alphaenumerate}
\end{definition}
In other words, a \zr{} is a subsequence of 4 accesses with key values in relative order $3,1,4,2$ and such that the axis-aligned rectangle that they span is empty (see Figure~\ref{fig:zrect-formal} for an example). We define the corresponding quantity, which we will prove is equivalent to the \FB{}.
\begin{definition}[\zr{} bound]
For any point set $P$ with distinct $x$- and $y$-coordinates,\footnote{If the $x$- and $y$-coordinates are not distinct, $\zRects(P)$ may give absurd results. For example, if we start with any $P$ and add a duplicate point $(x,y+\epsilon)$ for every point $(x,y)$ of $P$ (with $\epsilon$ small enough), then $\zRects(P)$ will drop to 0.} let
\[\zRects(P) \coloneqq |\{(p,q,r,s) \mid \text{$(p,q,r,s)$ is a \zr{} of $P$}\}|.\]
\end{definition}

First, we formally state the rotation-invariance of \zr{}s.
\begin{definition}[counter-clockwise 90\textdegree{} rotation]
For a point $p \in \R^2$, let $\rot{p} \coloneqq (-p.y,p.x)$.
Analogously, for a point set $P$, let $\rot{P} \coloneqq \{\rot{p} \mid p \in P\}$.
\end{definition}
\begin{lemma}\label{lemma:symmetry-zrects}
For any point set $P$, $\zRects(P) = \zRects(\rot{P})$.
\end{lemma}
\begin{proof}
Each \zr{} of $P$ induces a \zr{} in $\rot{P}$ and vice-versa: \zr{} $(p,q,r,s)$ in $P$ becomes \zr{} $(\rot{s},\rot{p},\rot{q},\rot{r})$ in $\rot{P}$ (the reader is encouraged to physically rotate the page containing figure \ref{fig:zrect-formal} in order to convince themselves of this fact). Therefore, $P$ and $\rot{P}$ have the same number of \zr{}s.
\end{proof}

The relation between $\W(P)$ and $\zRects(P)$ is proved in the following two lemmas.
\begin{lemma}\label{lemma:zrects-geq-funnel}
$\zRects(P) \geq \W(P)/2 - O(m)$.
\end{lemma}
\begin{lemma}\label{lemma:funnel-geq-zrects}
$\W(P) \geq 2\cdot \zRects(P)$.
\end{lemma}
We will use the fact that $P$ has distinct $x$- and $y$- coordinates.
\begin{proof}[Proof of Lemma~\ref{lemma:zrects-geq-funnel}]
We will show that for each $p \in P$, the funnel of $p$ induces at least $\floor{f(P,p)/2}-1$ different \zr{}s of the form $(p, \cdot, \cdot, \cdot)$. Summing this up for each $p$ then completes the proof.

Let's assume $f(P,p) \geq 4$; otherwise the claim holds vacuously. Let's number the points in $\FL \cup \FR$ (the funnel of $p$) by increasing $y$-coordinate as $a_1,a_2,\ldots,a_l$. Note that $l$ may be greater than $f(P,p)$, because a sequence of funnel points that are all on the same side of $p$ counts only for 1 in $f(P,p)$.

We will call $(i,j) \in [l]^2$ a \emph{left-straddling pair} if $i+1 < j$, $a_i.x > p.x$ and $a_j.x > p.x$, but for all $i < k < j$, $a_k.x < p.x$. That is, $a_i$ and $a_j$ are to the right of $p$ but all funnel points between them in order of height are to the left of $p$. Because funnel points alternate $f(P,p)-1$ times between the left and the right of $p$, there must be at least $\floor{f(P,p)/2}-1$ left-straddling pairs.

We claim that if $(i,j)$ is a left-straddling pair, then $(p,a_{i+1},a_i,a_j)$ is a \zr{}. Since all left-straddling pairs have distinct $i$, this produces $\floor{f(P,p)/2}-1$ distinct \zr{}s.

First, we verify that $p,a_{i+1},a_i,a_j$ have the correct relative positions. The order in $y$-coordinate is correct by definition of the numbering $a_1, \ldots, a_k$. For the order in $x$-coordinates, we know that $a_{i+1}$ is to the left of $p$ and $a_i, a_j$ are to its right, so we only need to verify that $a_i.x < a_j.x$. This is true because $a_i$ is in the funnel of $p$, so $\rect{pa_i}$ must be empty. If $a_i.x > a_j.x$, then $a_j$ would be in $\rect{pa_i}$.

What we still need to prove is that rectangle $[a_{i+1}.x,a_j.x] \times [a_i.y,p.y]$ is empty (except for points $p,a_{i+1},a_i,a_j$ themselves). First, since $a_i$, $a_{i+1}$ and $a_j$ in the funnel of $p$, we know that $\rect{pa_i}$, $\rect{pa_{i+1}}$ and $\rect{pa_j}$ are empty. This covers the zones pictured in Figure~\ref{fig:zrect-cover}.

\begin{figure}[h]
\centering

\newcommand{\access}[2]{\node[cross] (p#2) at (#1,#2) {};}
\begin{tikzpicture}[scale=0.5]
\draw[gray, pattern=north west lines, pattern color=gray] (2,4) rectangle (1,2);
\draw[gray, pattern=north east lines, pattern color=gray] (2,4) rectangle (3,1);
\draw[gray, pattern=north west lines, pattern color=gray] (2,4) rectangle (4,3);
\draw[gray] (1,1) rectangle (4,4);
\access{3}{1}
\access{1}{2}
\access{4}{3}
\access{2}{4}
\node[above=0mm of p4] {$p$};
\node[left=0mm of p2] {$a_{i+1}$};
\node[below=0mm of p1] {$a_i$};
\node[right=0mm of p3] {$a_j$};
\end{tikzpicture}
\caption{Proposed \zr{} $(p,a_{i+1},a_i,a_j)$ with empty rectangles $\rect{pa_{i}}$, $\rect{pa_{i+1}}$ and $\rect{pa_j}$ highlighted. If in addition we can prove that $\rect{a_ia_{i+1}}$ and $\rect{a_ia_j}$ are empty, then this is a valid \zr{}.}\label{fig:zrect-cover}
\end{figure}
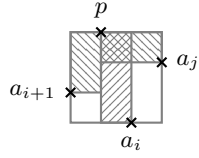

Finally, we will prove that $\rect{a_ia_{i+1}}$ and $\rect{a_ia_j}$ are empty, which covers the missing parts.
\begin{itemize}
\item Assume $\rect{a_ia_{i+1}}$ is not empty, and let $b$ be the highest point of $P$ in it (except for $a_{i+1}$).  We have already shown that $\rect{pa_i}$ and $\rect{pa_{i+1}}$ are empty, so $\rect{pb}$ must be empty. This means that $b$ must be in the funnel of $p$. But $a_i.y < b.y < a_{i+1}.y$, so this contradicts the numbering by increasing $y$-coordinate.
\item Assume $\rect{a_ia_j}$ is not empty, and let $b$ be the highest point of $P$ in it (except for     $a_j$). We have already shown that $\rect{pa_i}$ and $\rect{pa_j}$ are empty, so $\rect{pb}$ must be empty. This means that $b$ must be in the (right) funnel of $p$. But this contradicts our assumption that all funnel points between $a_i$ and $a_j$ in $y$-coordinate must be to the left of $p$.
\end{itemize}
Since points $p,a_{i+1},a_i,a_j$ and $[a_{i+1}.x,a_j.x] \times [a_i.y,p.y]$ is empty, $(p,a_{i+1},a_i,a_j)$ is a \zr{}. This completes the proof of Lemma~\ref{lemma:zrects-geq-funnel}.
\end{proof}

\begin{proof}[Proof of Lemma~\ref{lemma:funnel-geq-zrects}]
Essentially, the reason why this is true is because all \zr{}s must be exactly of the form described in the previous proof.
We will prove something slightly weaker which still reaches the desired result. We will group the \zr{}s by their top point and show that if $P$ has $k$ rectangles of the form $(p,\cdot,\cdot,\cdot)$, then $f(P,p) \geq 2k$.

Fix $p$, and sort the $k$ \zr{}s by the increasing $y$-coordinate of their bottom point $r$. Name their points $(p,q_1,r_1,s_1)$ to $(p,q_k,r_k,s_k)$. First, we will show that there can be no ties. Indeed, if $r_i.y = r_j.y$ then $r_i=r_j$. Also, when the $p$ and $r$ (top and bottom) points of a \zr{} are fixed, then the other two points $q$ and $s$ are uniquely determined as the rightmost point in $(-\infty,p.x] \times [r.x,p.x]$ and the leftmost point in $[r.x,\infty) \times [r.x,p.x]$, respectively.

We will now prove that
\begin{equation}\label{eq:alternating-qs}
q_1.y < s_1.y < q_2.y < s_2.y < \cdots < q_k.y < s_k.y.
\end{equation}
The $q_i.y < s_i.y$ inequalities are true by the definition of a \zr{}, so we only need to prove $s_i.y < q_{i+1}.y$.
To do this, consider two consecutive \zr{}s $(p,q_i,r_i,s_i)$ and $(p,q_{i+1},r_{i+1},s_{i+1})$ (see Figure~\ref{fig:two-zrects}). Since $r_i.y < r_{i+1}.y$, $s_i$ can't be strictly to the right of $r_{i+1}$, because otherwise $r_{i+1}$ would be inside \zr{} $(p,q_i,r_i,s_i)$. In turn, this means that $s_i$ can't be strictly higher than $r_{i+1}$ because otherwise it would be inside $\rect{pr_{i+1}}$. Therefore, we have $s_i.y \leq r_{i+1}.y < q_{i+1}.y$.

\begin{figure}[h]
\centering

\newcommand{\access}[2]{\node[cross] (p#2) at (#1,#2) {};}
\begin{tikzpicture}[scale=0.5]
\draw[gray] (1,4) rectangle (8,7);
\draw[gray] (2,1) rectangle (5,7);
\access{4}{1}
\access{2}{2}
\access{5}{3}
\access{7}{4}
\access{1}{5}
\access{8}{6}
\access{3}{7}
\node[above=0mm of p7] {$p$};
\node[left=0mm of p2] {$q_i$};
\node[left=0mm of p5] {$q_{i+1}$};
\node[below=0mm of p1] {$r_i$};
\node[below=0mm of p4] (dup2) {$r_{i+1}$};
\node[right=0mm of p3] (dup1) {$s_i$};
\node[right=0mm of p6] {$s_{i+1}$};
\node[right] (text) at (6,1) {could be the same point};
\draw[->] (text) -- (dup1);
\draw[->] (text) -- (dup2);
\end{tikzpicture}
\caption{The only possible relative position of two \zr{} with the same top point $p$}\label{fig:two-zrects}
\end{figure}
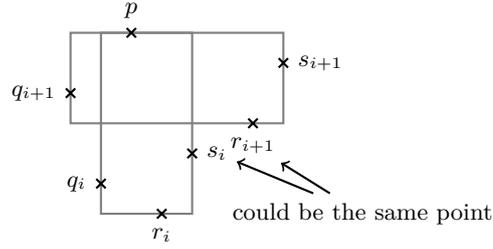

Points $q_1, s_1, \ldots, q_k, s_k$ are all in the funnel of $p$ by the definition of \zr{}. Therefore,  Equation~\eqref{eq:alternating-qs} reveals $2k$ funnel points that alternate from the left to the right side of $p$ with increasing $y$-coordinates. Thus $\mix(\FL.y, \FR.y)$ contains a subsequence $\mathtt{LRLR\cdots LR}$ of length $2k$, and
\[f(P,p) = \blocks(\mix(\FL.y, \FR.y)) \geq \blocks(\underbrace{\mathtt{LRLR\cdots LR}}_\text{length $2k$}) = 2k.\]
Summing this up for each $p$ completes the proof.
\end{proof}

\begin{corollary}\label{cor:funnel-reverse}
$\W(P) \geq \W(\rev{P}) - O(m)$.
\end{corollary}

\begin{proof}
By the left-right symmetry of $\W(\cdot)$, we know that $\W(\rev{P}) = \W(\rotrot{P})$, where $\rotrot{P}$ is $P$ rotated by 180\textdegree. Therefore,
\begin{align*}
\W(P)
&\geq 2\cdot\zRects(P)\tag{Lemma~\ref{lemma:funnel-geq-zrects}}\\
&= 2\cdot\zRects(\rotrot{P})\tag{Lemma~\ref{lemma:symmetry-zrects}}\\
&\geq \W(\rotrot{P}) - O(m)\tag{Lemma~\ref{lemma:zrects-geq-funnel}}\\
&= \W(\rev{P}) - O(m).
\end{align*}
\end{proof}

We can now finally prove Theorem~\ref{thm:domination}.
\begin{proof}[Proof of Theorem~\ref{thm:domination}]
By Lemma~\ref{lemma:two-sided},
$\ALT_\cT(P) \leq \W(P) + \W(\rev{P})$. Combining this with Corollary~\ref{cor:funnel-reverse}, we obtain $\ALT_\cT(P) \leq \W(P) + (\W(P) + O(m)) \leq O(\W(P) + m)$.
\end{proof}

\section{Separation between the \IB{} and the \FB{}}

\newcommand{\sL}{s_\ttL}
\newcommand{\sR}{s_\ttR}
\newcommand{\sLeft}{\sL}
\newcommand{\sRight}{\sR}
We will now define an access sequence $\xTilde$ such that the \IB{} is too low for all reference trees $\cT$ simultaneously. More precisely, we will define an access sequence $\xTilde \in [n]^m$ such that $\ALT_\cT(\xTilde) = O(m)$ for all trees $\cT$ while on the other hand $\OPT(\xTilde)$ and $\W(\xTilde)$ are $\Theta(m\lg\lg n)$. This $\lg\lg n$ factor is the biggest possible separation: indeed, Tango trees show that for a balanced tree $\cT$, $\ALT_\cT(\x)$ is always within $O(\lg \lg n)$ of $\OPT(\x)$.

To define $\xTilde$, we will need the notion of a \emph{bit-reversal} sequence. This is a permutation that in a sense looks ``maximally shuffled'' to a binary search tree.
\newcommand{\bitrev}{\mathrm{bitReversal}}
\begin{definition}
Let $k$ be a positive integer and let $K=2^k$. Then let $\bitrev^k \in \{0,\cdots,K-1\}^K$ be the sequence where $\bitrev^k_i$ is the number obtained by taking the binary representation of $i-1$, padding it with leading zeroes to reach length $k$, flipping it, then converting this back to a number.
\end{definition}
It is easiest to understand through an example. Take $k=2$, then $\bitrev^2$ is obtained this way:
\[(0,1,2,3) \xrightarrow{\text{to binary}} (00,01,10,11) \xrightarrow{\text{flip}} (00,10,01,11) \xrightarrow{\text{from binary}} (0,2,1,3).\]

The reason why we use this sequence is the following well-known fact.
\begin{fact}\label{fact:bitrev}
Let $\cT$ be the complete binary tree of height $k$ which has $K$ leaves labeled $0$ through $K-1$. Then $\ALT_\cT(\bitrev^k)=kK = K\lg K$.
\end{fact}
\begin{proof}
Because of the way $\bitrev^k$ is defined, for each node $u \in \cT$, the keys that are accessed below $u$ as the sequence is processed constantly alternate from $u$'s left subtree to $u$'s right subtree. So the contribution of $u$ is exactly the number of keys of its subtree. This way, every key is counted once at each of the $k=\lg K$ levels, so the total is $K \lg K$.
\end{proof}

We can now define our access sequence as follows. Let $n\coloneqq 2^K = 2^{2^k}$, and let
\[S_i \coloneqq (i + 2^{\bitrev^k_1}, i + 2^{\bitrev^k_2}, \ldots, i+2^{\bitrev^k_K}).\]
Then, denoting concatenation by $\circ$, we define
\[\xTilde \coloneqq \underbrace{S_0 \circ \cdots \circ S_0}_\text{$n$ times} \circ \underbrace{S_1 \circ \cdots \circ S_1}_\text{$n$ times} \circ \cdots \circ \underbrace{S_{n/2} \circ \cdots \circ S_{n/2}}_\text{$n$ times}.\]
The range of $\xTilde$ is $[n]$ and its length is $m = (\frac{n}{2}+1) \cdot n \cdot K = \Theta(n^2 \lg n)$.
See Figure~\ref{fig:separation-seq} for an example with $k=2$.
We will prove that for all $\cT$, $\ALT_\cT(\xTilde) \leq O(m)$ while on the other hand $\W(\xTilde) \geq \Omega(m \lg\lg n)$.
\begin{lemma}\label{lemma:alt-weak}
For any $\cT$, $\ALT_\cT(\xTilde) \leq O(m)$.
\end{lemma}

\begin{lemma}\label{lemma:funnel-strong}
$\W(\xTilde) \geq \Omega(m \lg\lg n)$.
\end{lemma}

The combination of Lemma~\ref{lemma:alt-weak} and Lemma~\ref{lemma:funnel-strong} shows the separation claimed in Theorem~\ref{thm:separation}. 
Before we move to the proofs of those lemmas, let's go over some intuition for the proof of Lemma~\ref{lemma:alt-weak}, which is the more complicated one.

First, note that the only reason we use $\bitrev^k$ in $\xTilde$ is to make $\W(\xTilde)$ large. Replacing $\bitrev^k$ by any other permutation of $\{0,\ldots,K-1\}$ would not affect the proof of Lemma~\ref{lemma:alt-weak} in any way because that proof only looks at the \emph{set} of keys that are hit by each of the parts $S_0, \ldots, S_{n/2}$.

The general intuition of the proof of Lemma~\ref{lemma:alt-weak} is that while one tree could give a high lower bound for \emph{one} of the sequences $S_i$, no tree can give a high lower bound \emph{on average} over all $S_i$. The reason is that, given the geometric spacing of each $S_i$, any way to split an interval of keys into two will typically (on average over $i$) leave almost all the keys of $S_i$ in either the left or the right part (Claim~\ref{claim:min-li-ri}). Therefore, it is impossible to split the keys into subtrees in a way that would ensure a high number of alternations.

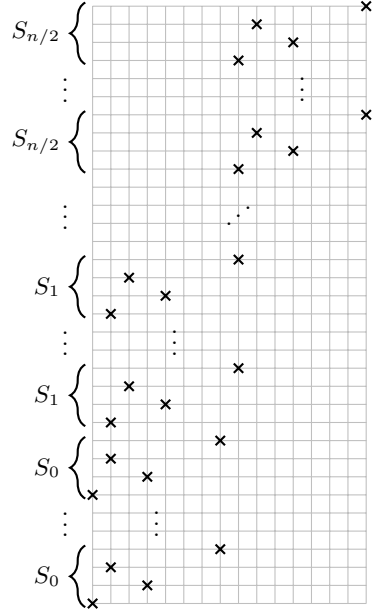
\begin{figure}[h]
\centering
\newcommand{\access}[2]{\node[cross] (p#2) at (#1,#2) {};}
\newcommand{\Szero}[3]{
\access{1+#1}{1+#2}
\access{4+#1}{2+#2}
\access{2+#1}{3+#2}
\access{8+#1}{4+#2}
\draw [decorate,decoration={brace,amplitude=2mm},xshift=-4mm,yshift=0pt]
(1,0.8+#2) -- (1,4.2+#2) node [midway,left,xshift=-2mm] {#3};
}
\newcommand{\customDots}[3]{
\node at (-0.5,1.8+#2) {$\vdots$};
\node at (4.5+#1,1.85+#2) {#3};
}
\newcommand{\group}[3]{
\Szero{#1}{#2}{#3}
\customDots{#1}{#2+4}{$\vdots$}
\Szero{#1}{#2+6}{#3}
}

\begin{tikzpicture}[scale=0.3]
\draw[step=1cm,lightgray,thin] (1,1) grid (16,34);
\group{0}{0}{$S_0$}
\group{1}{10}{$S_1$}
\customDots{4.5}{21}{\reflectbox{$\ddots$}}
\group{8}{24}{$S_{n/2}$}
\end{tikzpicture}

\caption{A schematic view of sequence $\xTilde$ for $k=2$. Each part $S_i$ is made of $K=2^k=4$ accesses. There are $n=2^K=16$ distinct keys and the length of $\xTilde$ is $m=(16/2+1)nK=576$.}\label{fig:separation-seq}
\end{figure}

\begin{proof}[Proof of Lemma~\ref{lemma:alt-weak}]
The first step of the proof is to decompose $\xTilde$ into substrings $S_0 \circ \cdots \circ S_0$ through $S_{n/2} \circ \cdots \circ S_{n/2}$, and then bound the sum of their Alternation bounds. Let's denote those substrings as $S_0 * n$, $S_1 * n$, \ldots, $S_{n/2}*n$. Because of the subadditivity of $\mixValue$ under concatenation (Fact~\ref{fact:props-mixvalue}), we have
\begin{equation}\label{eq:deconcatenate}
\ALT_\cT(\xTilde) \leq \sum_{i=0}^{n/2} \ALT_\cT(S_i * n).
\end{equation}

Note that we don't want to decompose $\xTilde$ down to the $S_i$'s themselves: every time we split it, our analysis loses up to an additive $O(n)$ in precision. Intuitively, this $O(n)$ is due to a ``warmup'' cost which we might or might not incur at the beginning of each substring, depending on which parts of the tree were last visited. With our decomposition into $n$ substrings, that's an extra $O(n^2)$ cost, which is okay since it is small compared to the total length of the sequence $\Theta(n^2 \log n)$. In fact, this is precisely why we repeated each $S_i$ several times: if we had defined $\xTilde$ as $S_0 \circ S_1 \circ \cdots \circ S_{n/2}$ instead, this $O(n^2)$ would have been large compared to the length of the sequence $\Theta(n \log n)$.\footnote{The astute reader will notice that we could have repeated each $S_i$ only $\Theta(n / \log n)$ times instead of $n$ times. But we are not limited in terms of the length of $\xTilde$, so it was (notationally) simpler to repeat them $n$ times.}

We will upper-bound the sum $\sum_i \ALT_\cT(S_i * n)$ by induction on the recursive definition of $\ALT_\cT(\cdot)$. 
Concretely, let $\cT^*$ be a subtree of $\cT$, and let $\TL^*,\TR^*$ be the left and right subtrees of $\cT^*$. Let $s$, $\sL$ and $\sR$ be the number of keys in $\cT^*$, $\TL^*$ and $\TR^*$ (note that $s=\sL+\sR$). For each $i$, let $P_i^*$ be the subset of $P(S_i * n)$ corresponding to keys in $\cT^*$, and let $P_{i,\ttL}^*,P_{i,\ttR}^*$ be the same for $\TL^*$ and $\TR^*$. We will prove the following claim by induction:
\begin{claim}\label{claim:induction-log}
For some constant $C>0$,
\[\sum_{i=0}^{n/2} \ALT_{\cT^*}(P_i^*) \leq (s-1)(n/2+1) + 2Cns \lg s.\]
\end{claim}

The base case is when $\cT^*$ is a single node. Then $\ALT_{\cT^*}(S_i * n)=0$ for all $i$, while $s=1$, so the result holds. 
To deal with the inductive step, we will need make a few tools first.
By definition of the \IB{} (Definition~\ref{def:ib}), for each $i$ we have
\begin{equation}\label{eq:apply-def-ib}
\ALT_{\cT^*}(P_i^*) = a(P_i^*,\cT^*) + \ALT_{\TL^*}(P_{i,\ttL}^*) + \ALT_{\TR^*}(P_{i,\ttR}^*).
\end{equation}
The challenging part is how to deal with $a(P_i^*,\cT^*)$. By Fact~\ref{fact:props-mixvalue}, we have
\[a(P_i^*,\cT^*) = \mixValue(P_{i,\ttL}^*.y,P_{i,\ttR}^*.y) \leq 2\cdot\min(|P_{i,\ttL}^*|,|P_{i,\ttR}^*|) + 1.\]
Summing this up over all $i$, we get
\begin{equation}\label{eq:leq-min}
\sum_{i=0}^{n/2} a(P_i^*,\cT^*) \leq \sum_{i=0}^{n/2}(2\cdot\min(|P_{i,\ttL}^*|,|P_{i,\ttR}^*|) + 1) = (n/2 +1)+ 2\cdot\sum_{i=0}^{n/2} \min(|P_{i,\ttL}^*|,|P_{i,\ttR}^*|).
\end{equation}

\begin{claim}\label{claim:min-li-ri}
For some constant $C>0$,
\[
\sum_{i=0}^{n/2} \min(|P_{i,\ttL}^*|,|P_{i,\ttR}^*|) \leq Cn \cdot
\begin{cases}
\sLeft\,\lg\frac{s}{\sLeft}\text{ if $\sL\leq\sR$, and}\\
\sRight\,\lg\frac{s}{\sRight}\text{ if $\sR\leq\sL$.}
\end{cases}
\]
This left-right symmetry is very surprising given that the sequences $S_i$ themselves are \emph{not} left-right symmetric. But it will be very convenient.
\end{claim}
\begin{claimproof}
To simplify the notation, let's say that the keys in $\TL^*$ are in range $[a,b]$ and the keys in $\TR^*$ are in range $[b,c]$, for some real numbers $a,b,c$ with $b-a=\sL$ and $c-b=\sR$.\footnote{We can for example fix $a$ to the first key of $\TL^*$ minus $\frac{1}{2}$, $b$ to the last key of $\TL^*$ plus $\frac{1}{2}$, and $c$ to the last key of $\TR^*$ plus $\frac{1}{2}$.}

For each $i$, let $V_i = \{i+2^0, \ldots, i+2^{K-1}\}$ be the set of values that are hit by sequence $S_i$. Then $|P_{i,\ttL}^*|$ (resp. $|P_{i,\ttL}^*|$) is exactly $n$ times the number of elements of $V_i$ that are in $[a,b]$ (resp. $[b,c]$). Let's name this number of keys $l_i$ (resp. $r_i$).
We will instead prove that
\begin{align}
\sum_{i=0}^{n/2} \min(l_i, r_i) &\leq
O\left(\sLeft\,\lg\frac{s}{\sLeft}\right)\text{ if $\sL\leq\sR$, and}\label{eq:sl-leq-sr}\\
\sum_{i=0}^{n/2} \min(l_i, r_i) &\leq
O\left(\sRight\,\lg\frac{s}{\sRight}\right)\text{ if $\sR\leq\sL$.}\label{eq:sr-leq-sl}
\end{align}
Once this is proved, $C$ can be set to the maximum of the two constants hidden inside the $O(\cdot)$s. Those constants might be different since the reasonings leading to~\eqref{eq:sl-leq-sr} and~\eqref{eq:sr-leq-sl} are completely different.

We first make a general observation. Look at set $V_i=\{i+2^0, \ldots, i+2^j, \ldots,\}$ in increasing order.
Note that after $i+2^j$, all further elements are spaced by at least $2^j$.
In order for $\min(l_i,r_i)$ to be non-zero, we need to have at least two elements of $S_i$ in $[a,c]$: specifically, one in $[a,b]$ and one in $[b,c]$. But this means that $i+2^{j+1} \in [a,c]$ isn't acceptable for $j > \lg s$: indeed, the closest other point in $S_i$ is more than $s$ away, so it must be outside of $[a,c]$.
Therefore, in bounding $\sum\min(l_i,r_i)$, it is fine to imagine that the elements $i+2^{j+1}$ for $j > \lg s$ simply do not exist.

Let us now prove \eqref{eq:sl-leq-sr}. Assume $\sLeft \leq \sRight$. We split into two cases:
\begin{itemize}
\item ``Far'' case: $i < a - \sLeft$. Since $i$ is further from $[a,b]$ than its size $\sLeft$, this means that $[a,b]$ can only contain at most one point from $S_i$. So $l_i \leq 1$. Besides, that (potential) single point must have $j \leq 1+\lg s$ (see above) and $j \geq \lg\sLeft$ (because we have $i+2^j \geq a$). And of course, we have in addition that $i+2^j \in [a,b]$. Therefore, this limits the number of possible values of $i$ to at most $\sLeft(2+\lg s - \lg \sLeft)$, and since $l_i \leq 1$, this also limits the total contribution to $\sum \min(l_i, r_i)$.
\item ``Close to right'' case: $i \geq a - \sLeft$. Then we also have $i \geq b - 2\sLeft$. Since we need $l_i \neq 0$ to have some contribution, we must have $i < b$, so the total number of possible values of $i$ is limited to $2\sLeft$. Let's consider the values of $j$ such that $i+2^j$ can lie in $[b,c]$, the right part. We already know that $j \leq 1+\lg s$, but we have no lower limit, as $i$ could be very close to $b$. However, values of $j$ much smaller than $\lg \sLeft$ will be only for the few values of $i$ close enough to $b$.

More precisely, we study the contribution of each $j$ to $\sum r_i$ into two groups:
    \begin{itemize}
    \item $j \geq \lg \sLeft$: there are $2+\lg s - \lg \sLeft$ such values $j$, and there are $2\sLeft$ possible values of $i$, so the total contribution is at most $2\sLeft(2 + \lg s - \lg \sLeft)$.
    \item $j < \lg \sLeft$: as $j$ decreases, the number of acceptable values of $i$ decreases exponentially. The number of values of $i$ for which $i+2^j \in [b,c]$ for $j \leq \lg \sLeft - l$ is at most $\sLeft / 2^l$. Therefore, the overal contribution is at most $\sLeft + \sLeft/2 + \cdots \leq 2\sLeft$.
    \end{itemize}
\end{itemize}
All those quantities are upper bounded by $O(\sLeft(1+\lg(s/\sLeft)))$, which under the assumption $\sLeft \leq \sRight$, is also bounded by $O(\sLeft\lg(s/\sLeft))$.

We now prove \eqref{eq:sr-leq-sl} in a very similar way. Assume $\sLeft \leq \sRight$.
\begin{itemize}
\item ``Far'' case: $i < b - \sRight$. The argument is analogous to the ``far'' case for \eqref{eq:sl-leq-sr}, but considering $r_i$ this time. We obtain a contribution of at most $\sRight(2+\lg s - \lg \sRight)$.
\item ``Close to right'' case: $i \geq b - \sRight$. The argument is analogous to the ``close to right'' case for \eqref{eq:sl-leq-sr}, but with a distance of $\sRight$ instead of $2\sLeft$ this time. We obtain contributions of at most $\sRight(2 + \lg s - \lg \sRight)$ and $2\sRight$ for the two subcases.
\end{itemize}
All those quantities are upper bounded by $O(\sRight(1+\lg(s/\sRight)))$, which under the assumption $\sRight \leq \sLeft$, is also bounded by $O(\sRight\lg(s/\sRight))$.
\end{claimproof}

We are now ready to finish the induction step.
\begin{claimproof}[Proof of Claim~\ref{claim:induction-log}]
We define $C$ to be the same as in Claim~\ref{claim:min-li-ri}. We have
\begin{align*}
\sum_{i=0}^{n/2} \ALT_{\cT^*}(P_i^*)
&= \sum_{i=0}^{n/2} \big(a(P_i^*,\cT^*) + \ALT_{\TL^*}(P_{i,\ttL}^*) + \ALT_{\TR^*}(P_{i,\ttR}^*)\big)\tag{by \eqref{eq:apply-def-ib}}\\
&\leq \left(\sum_{i=0}^{n/2} a(P_i^*,\cT^*)\right) + (\sL-1)(n/2+1) + 2Cn\sL \lg \sL + (\sR-1)(n/2+1) + 2Cn\sR \lg \sR\tag{inductive hypothesis}\\
&\leq (n/2+1) + 2\cdot\sum_{i=0}^{n/2} \min(|P_{i,\ttL}^*|,|P_{i,\ttR}^*|)\\
&\qquad+ (\sL-1)(n/2+1) + 2Cn\sL \lg \sL + (\sR-1)(n/2+1) + 2Cn\sR\lg\sR \tag{by \eqref{eq:leq-min}}\\
&\leq (s-1)(n/2+1)  + 2Cn(\sL \lg \sL + \sR\lg\sR) + 2\cdot\sum_{i=0}^{n/2} \min(|P_{i,\ttL}^*|,|P_{i,\ttR}^*|) \tag{$s=\sL+\sR$}
\end{align*}
All we need to show is that
\[Cn(\sL \lg \sL + \sR\lg\sR) + \sum_{i=0}^{n/2} \min(|P_{i,\ttL}^*|,|P_{i,\ttR}^*|) \leq Cn(s \lg s).\]
Let's assume that $\sL \leq \sR$ (the other case is identical). Then by Claim~\ref{claim:min-li-ri},
\begin{align*}
Cn(\sL \lg \sL + \sR\lg\sR) + \sum_{i=0}^{n/2} \min(|P_{i,\ttL}^*|,|P_{i,\ttR}^*|)
&\leq Cn(\sL \lg \sL + \sR\lg\sR) + Cn\sLeft\,\lg\frac{s}{\sLeft}\\
&\leq Cn(\sL \lg s + \sR\lg\sR)\\
&\leq Cn(\sL \lg s + \sR\lg s)\\
&= Cns \lg s.
\end{align*}
This completes the proof of Claim~\ref{claim:induction-log}.
\end{claimproof}

Applying Claim~\ref{claim:induction-log} to the full tree $\cT$, which has $n$ keys, we get
\begin{align*}
\ALT_\cT(\xTilde)
&\leq \sum_{i=0}^{n/2} \ALT_{\cT^*}(S_i * n)\tag{by \eqref{eq:deconcatenate}}\\
&\leq (n-1)(n/2+1) + 2Cn^2 \lg n\tag{Claim~\ref{claim:induction-log}}\\
&\leq O(n^2 \lg n)\\
&= O(m).
\end{align*}
\end{proof}

We now move to the proof of Lemma~\ref{lemma:funnel-strong}, which is much simpler.
\begin{proof}[Proof of Lemma~\ref{lemma:funnel-strong}]
From the definition of $\W(\cdot)$ (Definition~\ref{def:fb}), it is easy to see that for any two sequences $S$ and $T$, $\W(S \circ T) \geq \W(S) + \W(T)$. Indeed concatenating $S$ and $T$ does not affect the funnel of each point in $S$, and can only add points to the funnel of each point in $T$. Therefore,
\begin{equation}\label{eq:funnel-concat}
\W(\xTilde) \geq n\sum_{i=0}^{n/2} \W(S_i).
\end{equation}

Since $\W(\cdot)$ only depends on the relative order of the keys in the access sequence, not on their exact value, we have $\W(S_i) = \W(\bitrev^k)$ for each $i$. Besides, defining $\cT$ to be the complete binary search tree of height $k$ as in Fact~\ref{fact:bitrev}, we have
\begin{align*}
\W(\bitrev^k)
&\geq \Omega(\ALT_\cT(\bitrev^k)) - K\tag{by Theorem~\ref{thm:domination}}\\
&\geq \Omega(K \lg K) - K\tag{by Fact~\ref{fact:bitrev}}\\
&\geq \Omega(K \lg K).
\end{align*}
Combined with \eqref{eq:funnel-concat}, this gives $\W(\xTilde) \geq n \cdot (n/2+1) \cdot \Omega(K \lg K) \geq \Omega(m \lg K) = \Omega(m \lg \lg n)$.
\end{proof}

\ifarxiv
    \section{Towards an equivalence between the Funnel bound and the Independent Rectangle bound} \label{sec:speculation}
The Independent Rectangle bound $\IRB(P)$ of \cite{DHIKP09} is currently the highest known lower bound on 
$\OPT(P)$, as both the Alternation and Funnel bounds have been proven to be special cases of it. 
Nevertheless, in contrast to $\W(P)$, the quantity $\IRB(P)$ is complicated to analyze 
as it is a \emph{maximum} over a constrained family of lower bounds.
Therefore, proving that $\W(P)$ is actually 
equivalent to it (in accordance to Wilber's conjecture) could be very useful in analyzing
candidate optimal trees (e.g. GreedyFuture and splay trees). 
$\IRB(P)$ is equal (up to constant factors) to the sum $\upIRB(P) + \downIRB(P)$,\footnote{Actually, \cite{DHIKP09} uses $\upIRB(\cdot)$ and $\downIRB(\cdot)$ to refer to \emph{sets} of rectangles. Here, by $\upIRB(\cdot)$ and $\downIRB(\cdot)$ we actually refer to the size of those sets.} which are defined as the result of a sweeping line algorithm in point set $P$. No relationship is known between $\upIRB(P)$ and $\downIRB(P)$, but we conjecture that they are equal up an additive $O(m)$.

\begin{algorithm}[Algorithm 4.3 in \cite{DHIKP09}]\label{alg:up-greedy}
Sweep the point set $P$ with a horizontal line by increasing $y$-coordinate. When considering point $p$ on the sweep line, for each empty rectangle $\rect{pq}$ formed by $p$ and a point $q$ to its lower left, add the upperleft corner of $\rect{pq}$ to the point set. Let $\addUp(P)$ be the set of all added points (excluding the points originally in $P$), and let $\upIRB(P) \coloneqq |\addUp(P)|$.
\end{algorithm}

The set $\addDown(P)$ and quantity $\downIRB(P) \coloneqq |\addDown(P)|$ are defined in an analogous way, but considering $q$ to the lower \emph{right} of $p$ instead.
The following figure illustrates this process.
From now, we will make the distinction between \emph{accesses} (points of $P$, drawn as crosses) and \emph{added points} (points of $\addUp(P)$ or $\addDown(P)$, drawn as dots). See Figure~\ref{fig:irb} for an example of the computation of $\addUp(P)$ and $\addDown(P)$.

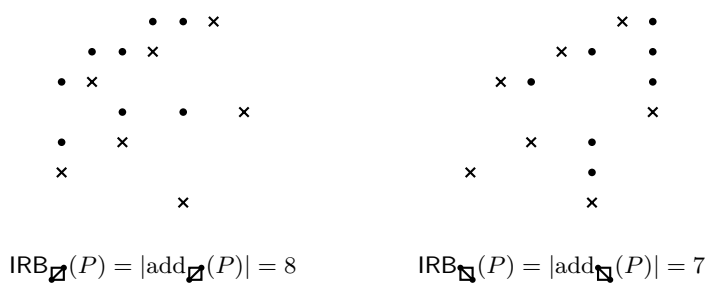
\begin{figure}[h]
\centering

\newcommand{\access}[2]{\node[cross] (p#2) at (#1,#2) {};}
\newcommand{\added}[2]{\node[dot] (q#2) at (#1,#2) {};}
\setlength{\tabcolsep}{8mm}
\begin{tabu}{cc}
\begin{tikzpicture}[scale=0.5]
\access{4}{0}
\access{0}{1}
\access{2}{2}
\access{6}{3}
\access{1}{4}
\access{3}{5}
\access{5}{6}
\added{0}{2}
\added{2}{3}
\added{4}{3}
\added{0}{4}
\added{1}{5}
\added{2}{5}
\added{3}{6}
\added{4}{6}
\end{tikzpicture}&
\begin{tikzpicture}[scale=0.5]
\access{4}{0}
\access{0}{1}
\access{2}{2}
\access{6}{3}
\access{1}{4}
\access{3}{5}
\access{5}{6}
\added{4}{1}
\added{4}{2}
\added{2}{4}
\added{6}{4}
\added{4}{5}
\added{6}{5}
\added{6}{6}
\end{tikzpicture}\\[2mm]
$\upIRB(P) = |\addUp(P)| = 8$&
$\downIRB(P) = |\addDown(P)| = 7$\\
\end{tabu}

\caption{Running Algorithm~\ref{alg:up-greedy} (and the analogous algorithm on the right) to compute $\protect\upIRB(P)$ and $\protect\downIRB(P)$}\label{fig:irb}
\end{figure}

\begin{remark}\label{rem:alg-empty}
As shown in \cite{DHIKP09}, all points $r$ in $\addUp(P)$ correspond to empty rectangles of $P$ in the following way. Let $a$ be the highest access of $P$ below $r$ such that $r.x=a.x$, and let $b$ be the access of $P$ such that $r.y=b.y$. Then $\rect{ab} \cap P = \{a,b\}$. In other words, $a$ is in the left funnel of $b$ (Definition~\ref{def:lr-funnel}).
\end{remark}

In this section, we prove that when $P$ contains only a constant number of \zr{}s, then $\addUp(P)$ is linear in $m$, or more precisely:
\begin{theorem}\label{thm:irb-linear}
For any point set $P$ with distinct $x$- and $y$-coordinates,
\[\upIRB(P) \leq O(m) + m \cdot \zRects(P).\]
\end{theorem}
Note that in case $\W(P)$ matches $\IRB(P)$, which is strongly believed to be true, then the statement could be improved to
\[\upIRB(P) \leq O(m + \zRects(P)).\]
Nevertheless, the current theorem is good news for the possible optimality of the Funnel bound.
The proof is a straightforward charging argument, and is a consequence of the following key lemma.
\begin{lemma}\label{lemma:empty-slab}
Let $a$ and $b$ be two points in the left funnel of $c$ , with $b$ to the upper left of $a$ ($a,b,c \in P$). Then either $P$ has no points in $[b.x,a.x] \times [c.y,\infty)$, or the lowest point in that region $d$ is part of a \zr{} of the form $(d,\cdot,\cdot,\cdot)$.
\end{lemma}
\begin{figure}[h]
\centering

\newcommand{\access}[2]{\node[cross] (p#2) at (#1,#2) {};}
\setlength{\tabcolsep}{10mm}
\begin{tabu}{cc}
\begin{tikzpicture}[scale=0.5]
\draw[lightgray] (5,1) rectangle (7,5);
\draw[lightgray] (1,4) rectangle (7,5);
\path[pattern=north west lines, pattern color=darkgray] (1,8) -- (1,5) -- (5,5) -- (5,8);
\draw[darkgray] (1,8.5) -- (1,5) -- (5,5) -- (5,8.5);

\access{5}{1}
\access{4}{2}
\access{2}{3}
\access{1}{4}
\access{7}{5}
\node[right=0mm of p5] (dup2) {$c$};
\node[below=0mm of p1] {$a$};
\node[left=0mm of p4] {$b$};
\node (text) at (9.5,7) {empty up to infinity};
\draw[->] (text) -- (5.1,7);
\draw[<->] (1,9) -- (5,9) node[midway,above] {$[b.x,a.x]$};
\end{tikzpicture}&
\begin{tikzpicture}[scale=0.5]
\draw[lightgray] (5,1) rectangle (7,5);
\draw[lightgray] (1,4) rectangle (7,5);
\draw[darkgray] (2,2) rectangle (6,7);
\access{5}{1}
\access{4}{2}
\access{2}{3}
\access{1}{4}
\access{7}{5}
\access{6}{6}
\access{3}{7}
\node[right=0mm of p6] (dup1) {$c'$};
\node[right=0mm of p5] (dup2) {$c$};
\node[above=0mm of p7] {$d$};
\node[below=0mm of p1] {$a$};
\node[left=0mm of p4] {$b$};
\node (text) at (8,8) {could be the same point};
\draw[->] (text) -- (dup1);
\draw[->] (text) -- (dup2);
\end{tikzpicture}\\
$[b.x,a.x] \times [c.y,\infty)$ is empty & $d$ is the top point of a \zr{}\\
\end{tabu}
\caption{The two cases of Lemma~\ref{lemma:empty-slab}. Rectangles $\rect{ac}$ and $\rect{bc}$ (in light gray) are empty.}\label{fig:empty-slab}
\end{figure}
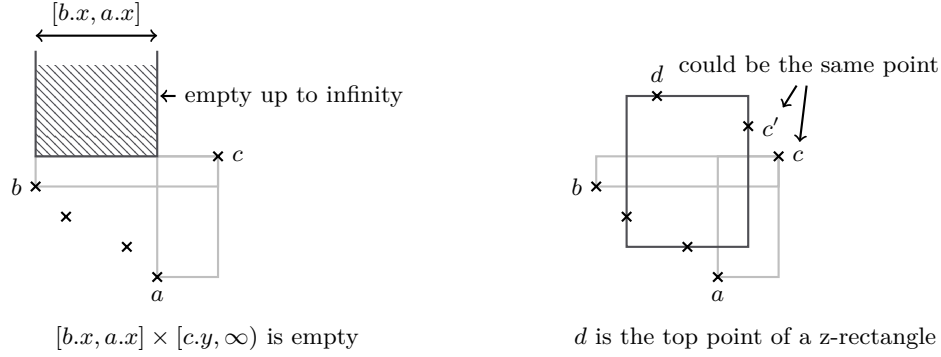

\begin{proof}
We start by proving this for $a$ and $b$ that are \emph{consecutive} left funnel points. That is, we assume that there is no point $a'$ in the left funnel of $c$ with $a.y < a'.y < b.y$. First, we observe that
\begin{equation}\label{eq:abc}
[b.x,c.x] \times [a.y,c.y] \cap P = \{a,b,c\}.
\end{equation}
Indeed, since $a,b$ are in the left funnel of $c$, we know that
\begin{itemize}
\item $[a.x,c.x] \times [a.y,c.y] \cap P = \{a,c\}$;
\item $[b.x,c.x] \times [b.y,c.y] \cap P = \{b,c\}$;
\end{itemize}
and besides, if there were a point in $[b.x,a.x] \times [a.y,b.y] \cap P$, then the highest of them would also be in the left funnel of $c$ and would contradict the consecutiveness of $a$ and $b$.

Now, assume that $P$ contains a point in $[b.x,a.x] \times [c.y, \infty)$ and let $d$ be the point among those with lowest $y$-coordinate. Let $c'$ be the point in $(a.x,\infty) \times [a.y,d.y]$ with least $x$-coordinate. Note that $c$ is an acceptable candidate, so $c'$ exists and $c'.x \leq c.x$.

The definitions of $d$ and $c'$ imply respectively that
\begin{itemize}
\item $[b.x,a.x] \times [c.y,d.y] \cap P = \{d\}$;
\item $(a.x,c'.x] \times [a.y,d.y] \cap P = \{c'\}$.
\end{itemize}
Therefore, combining those with \eqref{eq:abc}, we obtain that
\[[b.x,c'.x] \times [a.y,d.y] \cap P = \{a,b,c',d\}.\]

Also, again using \eqref{eq:abc} and the fact that $c'.x \leq c.x$, we can deduce that $c'.y \geq c.y > b.y$. Therefore, we have
\[b.x < d.x < a.x < c'.x\text{ and }a.y < b.y < c'.y < d.y\]
which means that $(d,b,a,c')$ is a \zr{}.

Now, suppose $a$ and $b$ are \emph{not} consecutive left funnel points, and let $a'_1,\cdots,a'_k$ be the left funnel points between them, by increasing $y$-coordinate (see Figure~\ref{fig:intermediate-funnel}). Then we can apply the above argument, replacing $(a,b)$ by each of $(a,a'_1)$, $(a'_1,a'_2)$, \ldots, $(a'_{k-1},a'_k)$ and $(a'_k,b)$. If $P$ has a point in $[b.x,a.x] \times [c.y,\infty)$, then the lowest such point $d$ will be in one of the ranges $[b.x,a'_k.x] \times [c.y,\infty), \ldots, [a'_1.x,a.x] \times [c.y,\infty)$, and thus will be involved in a \zr{} of the form $(d, \cdot, \cdot, \cdot)$.
\begin{figure}[h]
\centering

\newcommand{\access}[2]{\node[cross] (p#2) at (#1,#2) {};}
\begin{tikzpicture}[scale=0.5]
\draw[gray] (7,1) rectangle (9,6);
\draw[gray] (1,5) rectangle (9,6);

\access{7}{1}
\access{5}{2}
\access{4}{3}
\access{2}{4}
\access{1}{5}
\access{9}{6}
\node[right=0mm of p6] (dup2) {$c$};
\node[below=0mm of p1] {$a$};
\node[below=0mm of p2] {$a'_1$};
\node[below=0mm of p3] {$a'_2$};
\node[below=0mm of p4] {$a'_3$};
\node[below=0mm of p5] {$b$};
\end{tikzpicture}
\caption{Some intermediate points $a'_1, a'_2, a'_3$ in the left funnel of $c$ between $a$ and $b$}\label{fig:intermediate-funnel}
\end{figure}
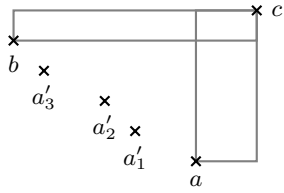

\end{proof}

The following lemma makes the charging argument concrete.
\begin{lemma}\label{lemma:highest-or-rightmost}
Every added point $p \in \addUp(P)$ is of at least one of three types:
\begin{alphaenumerate}
\item $p$ is the rightmost added point at $y$-coordinate $p.y$;
\item $p$ is the highest added point at $x$-coordinate $p.x$;
\item let $r$ be the lowest added point above $p$ at $x$-coordinate $p.x$, then $r$ has the same $y$-coordinate as some access $d \in P$ involved in a \zr{} $(d,\cdot,\cdot,\cdot)$.
\end{alphaenumerate}
See Figure~\ref{fig:charging-cases} for examples of each type.
\end{lemma}

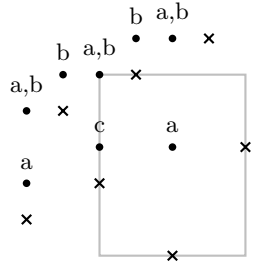
\begin{figure}[h]
\centering

\newcommand{\access}[2]{\node[cross] (p#2) at (#1,#2) {};}
\newcommand{\nodeAbove}[2]{\node[above=0mm of #1] {#2};}
\newcommand{\added}[3]{
    \node[dot] (q#1#2) at (#1,#2) {};
    \nodeAbove{q#1#2}{#3}
}
\begin{tikzpicture}[scale=0.6]
\draw[lightgray] (2,0) rectangle (6,5);
\access{4}{0}
\access{0}{1}
\access{2}{2}
\access{6}{3}
\access{1}{4}
\access{3}{5}
\access{5}{6}
\added{0}{2}{a}
\added{2}{3}{c}
\added{4}{3}{a}
\added{0}{4}{a,b}
\added{1}{5}{b}
\added{2}{5}{a,b}
\added{3}{6}{b}
\added{4}{6}{a,b}
\end{tikzpicture}

\caption{Added points of $\protect\addUp(P)$ labeled with their type(s) from Lemma~\ref{lemma:highest-or-rightmost}. The \zr{} corresponding to the type-c point is drawn in gray.}\label{fig:charging-cases}
\end{figure}

\begin{proof}
Consider the swipe of Algorithm~\ref{alg:up-greedy} when it reaches some access $c$. Let $p$ be any point added at this height ($p.y = c.y$). Assuming $p$ is not of type (a), there is another added point $q$ with $q.y = c.y$ and $q.x > p.x$. Let $q$ be the leftmost such point.

Let $a$ be the access at $x$-coordinate $q.x$ and $b$ be the access at $x$-cordinate $p.x$. Since all added points correspond to empty rectangles (Remark~\ref{rem:alg-empty}), we know that $a$ and $b$ are in the left funnel of $c$, with $a.y < b.y$. Thus we can apply Lemma~\ref{lemma:empty-slab}. There are two cases:
\begin{itemize}
\item Assume that there is no access in rectangle $[p.x,q.x] \times [c.y,\infty)$.
We claim this implies that $p$ is the highest added point at $x$-coordinate $p.x$, so $p$ is of type (b). Indeed, in order to produce a new added point at that $x$-coordinate, there would need, at some point later in the swipe, to be some access $d$ such that $\rect{dp}$ is empty. But since $d$ must be to the right of $q$, this is made impossible by the presence of $q$. 
\item Otherwise, let $d$ be the lowest access in rectangle $[p.x,q.x] \times [c.y,\infty)$. From Lemma~\ref{lemma:empty-slab}, we know that it is involved in a \zr{} of the form $(d,\cdot,\cdot,\cdot)$. Thus it suffices to prove the existence of $r$.
By the same arguments as the previous case, after it has added $p$ and $q$, Algorithm~\ref{alg:up-greedy} cannot add any points in range $[p.x,q.x]$ until it reaches $d$. Thus, when it reaches $d$, $\rect{dp}$ will be empty, which means that point $r=(p.x,d.y)$ will be added.
\end{itemize}
\end{proof}

\begin{proof}[Proof of Theorem~\ref{thm:irb-linear}]
Let's bound each type of added point as described in Lemma~\ref{lemma:highest-or-rightmost}.
By construction, the $y$-coordinates of any added point in $\addUp(P)$ has to be shared with one of the $m$ original accesses in $P$. Since that coordinate uniquely defines a point of type (a), there can be at most $m$ added points of type (a). An analogous argument can be made about $x$-coordinates to show that there are at most $n=m$ added points of type (b).

Furthermore, since there are $\zRects(P)$ \zr{}s, there are at most $\zRects(P)$ possible values of access $d$ in the definition of type (c). Each such $d$ can only produce $\leq m$ possible points $r$, and $r$ uniquely determines $p$. Therefore, there are at most $m \cdot \zRects(P)$ added points of type (c). Theorem~\ref{thm:irb-linear} follows from taking the sum over each type.
\end{proof}

\fi

\bibliography{refs}

\end{document}